\documentclass[sigconf,nonacm]{acmart}
\AtBeginDocument{%
  \providecommand\BibTeX{{%
    Bib\TeX}}}

\setcopyright{acmlicensed}
\copyrightyear{2027}
\acmYear{2027}
\acmDOI{XXXXXXX.XXXXXXX}
\acmConference[KDD '27]{the 33rd ACM SIGKDD Conference on Knowledge Discovery and Data Mining}{August 01--05,
  2027}{San Jose, California}
\acmISBN{978-1-4503-XXXX-X/2018/06}

\acmSubmissionID{1281}

\usepackage{tcolorbox}
\usepackage{xspace}

\usepackage{amsthm}
\usepackage{amsmath} 

\theoremstyle{definition}
\theoremstyle{plain}

\newtheorem{problem}{Problem}
\newtheorem{lemma}{Lemma}

\usepackage{algorithm}
\usepackage{algpseudocode}

\usepackage{booktabs}
\usepackage{multirow}

\usepackage{graphicx}

\usepackage{enumitem}

\usepackage{url}

\usepackage{hyperref}

\usepackage{xcolor}

\usepackage{enumitem}

\newcommand{\ProblemShort}{ROMSA\xspace}
\newcommand{\ProblemFull}{\textit{Robust Opinion Maximization via Stubbornness Adjustment}\xspace}
\newcommand{\ProblemBudget}{b-{\ProblemShort}\xspace}

\newcommand{\CustomLemmaThree}{\texttt{FastPoly}\xspace}
\newcommand{\CustomLemmaFour}{\texttt{FastEllip}\xspace}

\newcommand{\stubbo}{\ensuremath{\mathbf{a}}}
\newcommand{\stubbotarg}{\ensuremath{\mathbf{a}^\ast}}
\newcommand{\stubboInit}{\ensuremath{\mathbf{a}_{\text{init}}}}
\newcommand{\inner}{\ensuremath{\mathbf{s}}}
\newcommand{\innernoisy}{\ensuremath{\mathbf{\hat{s}}}}
\newcommand{\expressed}{\ensuremath{\mathbf{z}}}
\newcommand{\influence}{\ensuremath{\mathbf{q}}}
\newcommand{\influencetarg}{\ensuremath{\mathbf{q}^\ast}}
\newcommand{\influenceInit}{\ensuremath{\influence_{\text{init}}}}
\newcommand{\noise}{\ensuremath{\boldsymbol{\delta}}}
\newcommand{\polybounds}{\ensuremath{\boldsymbol{\rho}}}
\newcommand{\public}{\ensuremath{\zeta}}
\newcommand{\publicnoisy}{\ensuremath{\hat{\zeta}}}
\newcommand{\station}{\ensuremath{\boldsymbol{\pi}}}

\newcommand{\bigO}{\ensuremath{{\mathcal{O}}}}

\newcommand{\spara}[1]{\smallbreak\noindent{\bf{#1}}}

\newcommand{\ie}{{i.e.}}

\newcommand{\squishlist}{\begin{list}{$\bullet$}
  { \setlength{\itemsep}{0pt}
     \setlength{\parsep}{3pt}
     \setlength{\topsep}{3pt}
     \setlength{\partopsep}{0pt}
     \setlength{\leftmargin}{1.5em}
     \setlength{\labelwidth}{1em}
     \setlength{\labelsep}{0.5em} } }
\newcommand{\squishend}{
  \end{list}  }

\usepackage{amsmath, amsthm}
\newtheorem{theorem}{Theorem}

\usepackage{subcaption}

\usepackage{stfloats}

\begin{document}

\title{Robust Opinion Maximization}


\author{Nikos Theologis}
\affiliation{%
  \institution{University of Ioannina}
  \city{Ioannina}    
  \country{Greece}
}
\affiliation{%
  \institution{Archimedes/Athena Research Center}
  \city{Athens}    
  \country{Greece}
}
\email{n.theologis@athenarc.gr}

\author{Evaggelia Pitoura}
\affiliation{%
  \institution{University of Ioannina}
  \city{Ioannina}    
  \country{Greece}
}
\affiliation{%
  \institution{Archimedes/Athena Research Center}
  \city{Athens}    
  \country{Greece}
}
\email{pitoura@uoi.gr}

\author{Evimaria Terzi}
\affiliation{%
  \institution{Boston University}
  \city{Boston}  
  \country{USA}
}
\email{evimaria@bu.edu}

\author{Panayiotis Tsaparas}
\affiliation{%
  \institution{University of Ioannina}
  \city{Ioannina}    
  \country{Greece}
}
\affiliation{%
  \institution{Archimedes/Athena Research Center}
  \city{Athens}    
  \country{Greece}
}
\email{tsap@uoi.gr}

\begin{abstract}
Opinion-formation models are fundamental for analyzing social behavior and understanding social dynamics.
The Friedkin-Johnsen (FJ) model is the most commonly used model in this line of research.
In FJ, the expressed opinion of an individual evolves over time as the combination of the opinions of their peers and a fixed inner opinion, weighted by their stubbornness. 
A problem that has received considerable attention is the problem of \emph{Opinion Maximization}, where the goal is to adjust the parameters of the FJ model, such as stubbornness, or inner opinions, to maximize the average expressed opinion, i.e., the public opinion.
Most existing work assumes full knowledge of the parameters of the model, which is unrealistic. 
In this paper, we introduce the problem of \emph{Robust Opinion Maximization}, where we assume uncertainty for the inner opinions of the nodes, and we maximize the public opinion under worst-case assumptions on this uncertainty. 
Using tools from Robust Optimization, we formulate our problem as a max-min problem, and we develop an optimal two-step solution that is both highly intuitive and efficient.  
Experiments on several real-world networks demonstrate that our approach can achieve robust solutions for million-node graphs in minutes, and reveal that our robust strategies can significantly mitigate the impact of noise compared to standard baselines. 
\end{abstract}





\maketitle

\section{Introduction}\label{sec:intro}




As social media networks have become primary platforms for public debate, understanding how opinions form and spread is increasingly important. While opinion dynamics has long been studied~\cite{social-dynamics-survey}, the scale of modern digital platforms has renewed interest in these models as tools for analyzing large-scale social systems.

One of the most established frameworks in this area is the Friedkin and Johnsen (FJ) model \cite{Friedkin_Johnsen}. In this model, individuals are represented as nodes in a social network, each associated with an \emph{inner} and an \emph{expressed} opinion. The inner opinion is a fixed value capturing the inherent beliefs of an individual, while the expressed opinion evolves through social interactions. Specifically, individuals iteratively update their expressed opinions by balancing their inner beliefs against the opinions of their neighbors. This balance is governed by a \emph{stubbornness} parameter, which determines the extent to which an individual adheres to their inner belief versus incorporating the opinions of their peers.


The FJ model has found a wide range of applications in social network analysis~\cite{miyauchi2026survey}, including modeling campaigns~\cite{gionis},  polarization~\cite{matakos2017measuring,Musco,NEURIPS2021} and  conflict~\cite{chen2018quantifying, papachristou2025echoesdisagreementmeasuringdisparity}.
In this work, we focus on the problem of \emph{opinion maximization}~\cite{gionis}, where the goal is to  maximize the average expressed opinion, i.e., the public opinion, by adjusting the FJ parameters. 
Several variants of this problem have been considered that adjust the inner or expressed opinions~\cite{gionis,Opinion_Optimization,Wang2025}, the stubbornness parameters~\cite{Tsourakakis}, or the graph structure~\cite{ZHU2025115090}.
Existing work on opinion maximization assumes full knowledge of the model parameters.
However, in real-world scenarios, assuming 
precise data for these parameters is highly unrealistic due to measurement errors, hidden biases, or malevolent actors. 

To address this shortcoming, we introduce the novel \emph{Robust Opinion Maximization} problem.
Robust optimization \cite{bental1998robust,bertsimas2004price, bertsimasCaramanis2011Survey} studies how to solve optimization problems under worst-case scenarios for uncertain data. 
In our context, we assume uncertainty in the inner opinions, and aim to maximize the \emph{worst-case} public opinion, by adjusting the stubbornness values of the individuals.

To solve this problem, we propose a two-step algorithmic solution. In Step 1, we find the optimal \emph{influence vector}. 
We utilize the notion of \emph{influence}, as defined in \cite{theologis2026fairness} and by \citet{FriedkinRecent}, which determines the influence of the noisy inner opinions on the optimal worst-case public opinion.
In Step 2, we find the stubbornness values that achieve the optimal influence vector.
%

Using standard robust optimization techniques, we can formulate the Step 1 problem  as a Linear Program (LP) or a Second-Order Cone Program (SOCP), both solvable efficiently with off-the-shelf solvers. However, we show that, in our setting, the problem admits specialized $\mathcal{O}(n \log n)$ algorithms, leading to substantial improvements in scalability.
%
For the second step,
we derive a mapping between the influence vector and the stubbornness parameters, and we provide a linear-time algorithm for computing these parameters. 
Overall, the proposed two-step approach is highly scalable and is applicable even to large real-world networks.

In summary, 
we make the following contributions:
\squishlist
    \item 
    We define the novel problem of {\ProblemFull} (\ProblemShort). 
    We also define a budgeted version of the problem, {\ProblemBudget}, where we constrain the change in the input influence values. To the best of our knowledge we are the first to introduce robust optimization in opinion dynamics problems.
    \item 
    We design an optimal two-step solution, and we provide  scalable algorithms for each step, for both {\ProblemShort} and {\ProblemBudget}. For the first step, we propose $\bigO(n \log n)$ algorithms to calculate the optimal influence vector, and for the second step, we derive a linear-time methodology to compute stubbornness parameters given a target influence vector.
    \item 
    We conduct experiments on large-scale real-world networks that demonstrate the properties of our algorithms and the scalability of our approach.
\squishend
The rest of the paper is structured as follows. In Section~\ref{sec:related_work} we review related work.
In Section~\ref{sec:definitions}, we present the FJ model and the opinion maximization problem, we introduce the noise model, and we define the Robust Opinion Maximization problem. In Section~\ref{sec:algorithm}, we present our solution.  In Section~\ref{sec:budget}, we define the budgeted version of our problem. In Section~\ref{sec:experiments} we present our experimental evaluation. We conclude the paper in Section~\ref{sec:conclusion}.

\section{Related Work}
\label{sec:related_work}

To the best of our knowledge, we are the first to introduce robust optimization in opinion dynamics, define the \ProblemShort and \ProblemBudget problems, and design a scalable framework for their solution. 
However, our research is related to existing work on \emph{interventions} and \emph{uncertainty} in opinion-formation dynamics.

\spara{Interventions in opinion-formation models:}
There is a substantial body of work dedicated to modifying the parameters of the Friedkin-Johnsen (FJ) model to optimize specific social objectives. 
A recent survey \cite{miyauchi2026survey} provides a comprehensive overview. 
Specifically, works by \cite{Howbad, Musco, NEURIPS2021} focus on altering the network topology to minimize polarization or reduce overall social cost. 
By directly adjusting the opinions, \cite{gionis, matakos2017measuring} introduced the problems of maximizing public opinion and minimizing polarization. Recent works by \cite{Wang2025, Sun2025} have also focused on scaling these types of structural and opinion-based adjustments to massive networks. 
A line of research by \cite{Tsourakakis, Xu2022, theologis2026fairness} focuses on the study of stubbornness and how such adjustments can be used to optimize the public opinion and fairness. Our work extends this literature by optimizing stubbornness interventions against worst-case noise in the inner opinions.

\spara{Uncertainty in opinion dynamics:}
While robustness has been studied for discrete influence cascade models \cite{He2016, Chen2016}, no such work has been made for continuous opinion dynamics like the FJ model.
While most algorithmic interventions for continuous opinion-formation models assume perfect model knowledge, literature addressing such uncertainty remains surprisingly sparse.
An early attempt was made by \cite{chen2018quantifying}, who addressed unknown inner opinions by minimizing the conflict risks over all possible binary opinion assignments. 
Recently, several works have tackled polarization minimization under unknown inner opinions: \cite{chaitanya2024sdp} utilized semidefinite programming relaxations, \cite{cinus2025minimizing} treated it as a budgeted query problem to reconstruct missing opinions, and \cite{cinus2026online} employed multi-armed bandits to learn opinions sequentially from noisy feedback. 
Also \cite{cianfanelli2025optimal} investigated uncertain network topologies and \cite{Tu2023} tackled limited information via an adversarial framework. Crucially, none of these approaches employ \emph{robust optimization}.


\section{Definitions}
\label{sec:definitions}

In this section, we provide the necessary definitions used throughout the paper, and we define the Robust Opinion Maximization problem. We use lower case letters for scalars ($\theta$, $d$, $q_i$), lowercase \textbf{bold} letters for column vectors ($\mathbf{q}$, $\mathbf{s}$, $\boldsymbol{\delta}$), and uppercase letters for matrices ($A$, $W$, $Q$). We denote the identity matrix as $I$ and the column vector full of ones as $\mathbf{1}$.

\subsection{Opinion-Formation Model }
We assume a directed 
graph $G=(V,E)$ consisting of $|V|=n$ nodes and $|E|=m$ edges.
We consider the popular Friedkin and Johnsen (FJ) model~\cite{Friedkin_Johnsen}, where
each node $i \in V$ has a fixed \emph{inner opinion} $s_i \in [0,1]$  and an \emph{expressed opinion} $z_i$.
The former is fixed and a characteristic of the node itself. 
The latter is the result of the iterative opinion-formation process: 
\begin{align}
z^{(t)}_i = a_i s_i + (1-a_i)\sum_{j \in N_i} w_{ij} z_j^{(t-1)} \, .
\end{align}
In the above equation, $N_i = \{j \in V \mid (i,j) \in E\}$ is the neighborhood of node $i$ in $G$, and 
$w_{ij}$ quantifies the importance that node $i$ places on the expressed opinion of node $j$. 
The values of $w_{ij}$ 
define a \emph{row-stochastic} matrix $W \in [0,1]^{n \times n}$, with $w_{ij}>0$ for $(i,j)\in E$, and  
$\sum_{j=1}^n w_{ij}=1$, which we call the \emph{interaction matrix}. Finally, $a_i\in [0,1]$ corresponds to the \emph{stubbornness} of node $i$ and quantifies how determined node $i$ is to adhere to their own opinion.

Let \stubbo, {\inner} and {\expressed} denote the $n$-dimensional vectors of the stubbornness values, the inner and the expressed opinions of the nodes in $V$. We can 
write the update equation for the FJ model as:
\begin{align}\label{eq:dynamics}
{\expressed}^{(t)} = A{\inner} + (I-A)W{\expressed}^{(t-1)},
\end{align}
where
$A = \text{Diag}({\stubbo})$ is a diagonal matrix with $A[i,i]=a_i$. 
If  the  matrix $W$ is \emph{irreducible} ({\ie}, the underlying graph $G$ is strongly connected), and at least one node has stubbornness $a_i>0$, then 
vector $\expressed$ converges to a steady state where:
\begin{equation}
\label{eq:steady-state}
{\expressed} = \left(I-(I-A)W\right)^{-1}A{\inner}=Q\inner .
\end{equation}

We call $Q$ the \emph{influence matrix}, which is row-stochastic \cite{Friedkin_Johnsen, Howbad}, i.e., 
$Q[i,j]\in [0,1]$ and $\sum_{j \in V} Q[i,j] = 1$.
Note that $z_i = \sum_{j\in V} Q[i,j] s_j$, and thus, element $Q[i,j]$ determines the extent to which the expressed opinion $z_i$ of node $i$ is influenced by the inner opinion $s_j$ of node $j$.
We then define the \emph{influence}  $q_i = \frac{1}{n}\sum_{j \in V}Q[j,i]$ of node $i$ (and the corresponding \emph{influence vector} \influence) as the average influence node $i$ exerts on all the nodes in $V$.


Finally, we define
$\public = \frac 1n \sum_{i\in V}z_i$ as the \emph{public opinion} in the network.
Using the above definitions, we can see that
$\public = \sum_{i \in V} q_i s_i$, thus, $q_i$ is the influence of the inner opinion $s_i$ on the public opinion \public. 
We refer to the quantity $c_i = q_i s_i$ as the \emph{contribution} of the node $i$ to the public opinion. 
Note that due to the row-stochasticity of matrix $Q$ the influence vector satisfies $\sum_{i=1}^{n} q_i = 1$ and $q_i \geq 0, \forall i$. 
Thus, vector $\influence$ can be viewed as a \emph{distribution} of influence.
Influence is also defined as \emph{influence centrality} by \citet{FriedkinRecent}.

The public opinion $\public(\inner, \stubbo, $W$)$ is a function of the parameters of the FJ model. The \emph{Opinion Maximization} problem \cite{gionis} 
aims to maximize the public opinion, by modifying these parameters. Several variations of this problem have been considered in the literature \cite{Tsourakakis,Wang2025,ZHU2025115090}. In this work, we focus on the case where we adjust the stubbornness of the nodes. 

\subsection{Modeling Uncertainty}

Previous work on opinion maximization assumes that we have full knowledge of the FJ model parameters. The assumption that inner opinions $\inner$ are known exactly is a major limitation.
In real-world social network applications, inner opinions are often estimated through surveys, historical data, or proxy metrics, making them prone to estimation errors.
To capture this, we introduce uncertainty in the vector $\inner$. 
We assume the inner opinions are noisy and lie within an \emph{uncertainty set} $\mathcal{U}_s$ centered around the \emph{input inner opinion} vector $\inner \in [0,1]^n$. We adopt a \emph{multiplicative} uncertainty framework, and we define the \emph{noisy inner opinion} vector {\innernoisy} as:
\begin{equation}
    \innernoisy = \inner \circ (\mathbf{1} - \noise),
\end{equation}
where $\noise \in \mathbb{R}^n$ represents the vector of \emph{noise}, and $\circ$ denotes the element-wise product.
Each entry $\delta_i$ represents the \emph{percentage degradation} of the input inner opinion $s_{i}$. 
Because we will address an Opinion Maximization Problem \ref{problem:romsa}, the \emph{worst-case} realization of this noise corresponds strictly to a \emph{reduction} for the inner opinions. Thus, without loss of generality, the worst-case noise values will always fall within the range $\delta_i \in [0,1]$. 

Following classic robust optimization techniques \cite{bental1998robust, bertsimas2004price, bertsimasCaramanis2011Survey}, we consider two distinct models for the uncertainty set.

\spara{Polyhedral Uncertainty Set:}
This set is defined by the vector $\polybounds \in [0,1]^n$ and the scalar noise budget $\gamma \in [0,n]$:
\begin{align}
\label{noise poly}
\mathcal{U}_s^{\text{poly}} = \left\{ \innernoisy = \inner \circ (\mathbf{1} - \noise) \; \middle| \;
\begin{aligned}
    & 0 \le \delta_i \le \rho_i, & \forall i \in \{1, \dots, n\}, \\
    & \sum_{i=1}^{n} \delta_i \le \gamma
\end{aligned}
\right\}.
\end{align}

Parameter $\rho_i$ quantifies the confidence in the input inner opinion $s_{i}$ for each node $i$. The \emph{noise budget} $\gamma$ controls the aggregate reliability of the system by bounding the total sum of percentage losses across all nodes. Since $\delta_i \in [0,1]$, the maximum possible value for $\gamma$ is $n$. In the extreme case when $\polybounds=\mathbf{1}$, the noise budget $\gamma$ represents the maximum number of nodes whose inner opinions can be reduced to zero.

\spara{Ellipsoidal Uncertainty Set:}
This set is defined by the matrix ${\Sigma} \in \mathbb{R}^{n \times n}$ and the scalar noise budget $\omega \in \mathbb{R}_{\ge 0}$:
\begin{equation}
\label{noise_elips}
\mathcal{U}_s^{\text{ellip}} = \left\{ \innernoisy = \inner \circ (\mathbf{1} - \noise) \; \middle| \; \noise^T {\Sigma}^{-1} \noise \le \omega^2 \right\}.
\end{equation}

The diagonal entries of the matrix, $\Sigma [i,i] =\sigma_i^2$, represent the variance of $\delta_i$, with $\sigma_i \in [0,1]$ serving as a measure of the local confidence in the input inner opinion $s_{i}$. 
Furthermore, the off-diagonal entries $\Sigma [i,j]$ allow us to model \emph{correlations} between the noises, making it possible to model groups of nodes whose opinions are likely to shift together due to echo chambers.
Thus we have \emph{correlated} or \emph{uncorrelated} ellipsoidal uncertainty, with matrix $\Sigma$ being diagonal in the latter case.
Finally, to guarantee that worst-case noise {\noise} results in valid noisy opinions {\innernoisy} we bound the {noise budget} $\omega \in \left[0, \, 1 / \max_i \sigma_i \right]$. 
For more details about the ellipsoidal uncertainty set see Appendix \ref{appendix:covariance_setup}.


\subsection{The \ProblemShort Problem}
\label{sec:problem}

We now formally define the {\ProblemFull} (\ProblemShort) problem.
For noisy inner opinions, i.e., \innernoisy, we also have \emph{noisy expressed opinions} $\hat{z}_i = \sum_{j\in V} Q[i,j] \hat{s}_j$. 
Thus, we define the \emph{noisy public opinion} {\publicnoisy} as:
\begin{align}
    \publicnoisy(\stubbo, \innernoisy) = \frac 1n \sum_{i\in V}\hat{z}_i = \frac 1n \sum_{i\in V}\sum_{j\in V} Q[i,j] \hat{s}_j = \influence (\stubbo)^T \innernoisy,
\end{align}
Our goal is to find the stubbornness vector {\stubbo} that maximizes the noisy public opinion $\publicnoisy$, given the uncertainty set $\mathcal{U}_s$. Formally, we define the \ProblemShort problem as follows:

\begin{problem}[\ProblemFull -- \ProblemShort]
\label{problem:romsa}
Given a directed social network graph $G = (V, E)$ with interaction matrix $W$, a vector of input inner opinions $\inner \in [0,1]^n$, and the uncertainty set of the inner opinions $\mathcal{U}_s$, 
find a stubbornness vector $\stubbotarg \in [0,1]^n$ such that:
\begin{align} 
\label{ROMSA}
\stubbotarg = \arg\max_{\stubbo\in [0,1]^n} \min_{\innernoisy \in \mathcal{U}_s} \publicnoisy(\stubbo, \innernoisy)\, . 
\end{align}
\end{problem}

The \ProblemShort problem aims to maximize the public opinion in the network while guarding against worst-case noise in the inner opinions.
We choose stubbornness as our control variable because we view it as a feasible non-intrusive intervention, compared to altering the network topology, or the inner opinions.  
Intuitively, adjusting stubbornness implies empowering some individuals to be more assertive in their views (increasing $a_i$) while encouraging others to be more receptive to their social circle (decreasing $a_i$). 
Because it is physically interpretable and effective, adjusting stubbornness has emerged as a standard mechanism in the recent literature on the FJ model~\cite{Tsourakakis,Xu2022,theologis2026fairness}.

From the computational point of view, 
the maximization and minimization components of {\ProblemShort} (see Eq.~\ref{ROMSA}) are easy to solve in isolation. 
For the outer maximization problem, in the absence of uncertainty (the \emph{nominal solution} for  $\innernoisy=\inner$), the optimal strategy is to identify the node with the highest input inner opinion, $i^\ast = \arg\max_{i\in V} s_{i}$, and concentrate all network influence on this node by setting $a_{i^\ast} = 1$ and $a_j = 0$ for all $j \neq i^\ast$. 
Conversely, if the stubbornness vector $\mathbf{a}$ is fixed, the inner minimization reduces to an optimization over the uncertainty set $\mathcal{U}_s$. For both polyhedral and ellipsoidal uncertainty sets, this has a straightforward analytical or greedy optimal solution (see Appendix~\ref{appendix:inner_min}).
Indicatively, for the polyhedral noise, the worst case scenario is when all the noise budget is allocated to the nodes that have the highest \emph{contribution} $c_i = q_i s_i$ to the public opinion.

Taken together, however, these two components couple to form a complex problem. 
With the introduction of uncertainty, assigning total influence to node $i^\ast = \arg\max_i s_{i}$ is not necessarily a good strategy. In the worst case, the noise $\noise$ can severely reduce the public opinion by concentrating entirely on this single node $i^\ast$. 
In general, the worst-case noise 
depends on the influence distribution $\influence(\stubbo)$, 
``targeting'' nodes that are promoted by $\influence$.
Consequently, the {\ProblemShort} problem is an intricate balancing act. 
To maximize public opinion, we must concentrate influence on a few nodes with high inner opinions, 
while also diversifying enough to protect against worst-case noise.



\section{Solving the \ProblemShort Problem}
\label{sec:algorithm}

The \emph{max-min} structure of the \ProblemShort problem  is computationally challenging due to the non-convex relationship between the stubbornness vector $\stubbo$ and the influence distribution $\influence(\stubbo)$. 
To solve the \ProblemShort problem optimally in a tractable manner, we propose a two-step decoupling solution:

\squishlist
\item In \textbf{Step 1}, we 
treat the influence vector $\influence$ as our primary control variable
and find the optimal \influencetarg, which we call the \emph{target influence}; 
\item In \textbf{Step 2}, we find the stubbornness vector $\stubbotarg \in [0,1]^n$ that achieves this target influence $\influencetarg$.
\squishend

\subsection{Step 1: Finding the Optimal Target Influence}

In Step 1, we use {\influence} as our control variable, and we seek the optimal influence vector $\influencetarg$ that maximizes the worst-case noisy public opinion under uncertainty, $\publicnoisy(\influence) = \influence^T \innernoisy$. That is: 
\begin{align}\label{eq:transformation}
    \influencetarg = \arg\max_{\influence} \min_{\innernoisy \in \mathcal{U}_s} \influence^T \innernoisy.
\end{align}
Since the problem described by Eq.~\eqref{eq:transformation} is linear in $\influence$ and $\innernoisy$, we can apply standard robust-optimization techniques \cite{bental1998robust, bertsimas2004price, bertsimasCaramanis2011Survey} to reformulate the \emph{max-min} problem into a tractable maximization problem. 
These transformations lead to the following Lemma. 

\begin{lemma}[Step 1 Reformulations]
\label{lemma:reform}
For the polyhedral uncertainty set $\mathcal{U}_s^{\text{poly}}$, by introducing dual variables $\theta$ and $p_i$, the {max-min} problem defined by Eq.~\eqref{eq:transformation} is solved optimally by the following Linear Program (LP):
\begin{equation}
\label{eq:poly_robust_main}
\begin{aligned}
    & \max_{\influence, \theta, \boldsymbol{p}} \quad && \influence^T \inner - \gamma \theta - \sum_{i=1}^n \rho_i p_i \\
    & \text{\emph{subject to}} \quad && \theta + p_i \ge q_i s_{i}, \quad \forall i \in \{1, \dots, n\} \\
    & && \sum_{i=1}^n q_i = 1, \quad  q_i \ge 0, \quad \forall i \in \{1, \dots, n\} \\
    & && \theta \ge 0, \quad p_i \ge 0, \quad \forall i \in \{1, \dots, n\}
\end{aligned}
\end{equation}
For the ellipsoidal uncertainty set $\mathcal{U}_s^{\text{ellip}}$, by introducing an auxiliary variable $t$, the {max-min} problem defined by Eq.~\eqref{eq:transformation} is solved optimally by the following Second-Order Cone Program (SOCP):
\begin{equation}
\label{eq:ellip_robust_main}
\begin{aligned}
    & \max_{\influence, t} \quad && \influence^T \inner - \omega t \\
    & \text{\emph{subject to}} \quad && \left\| {\Sigma}^{1/2} \text{\emph{diag}}(\inner) \influence \right\|_2 \le t \\
    & && \sum_{i=1}^n q_i = 1, \quad  q_i \ge 0, \quad \forall i \in \{1, \dots, n\}
\end{aligned}
\end{equation}
\end{lemma}

The step-by-step derivations utilizing strong duality are provided in Appendix \ref{appendix:robust_proofs}. 
Solving the LP and SOCP defined in Lemma~\ref{lemma:reform} has worst-case $\bigO(n^3)$ complexity~\cite{Boyd2004}. In practice, it can be solved efficiently, even for large graphs, using off-the-shelf solvers. 
However, we can derive intuitive $\mathcal{O}(n \log n)$ algorithms for computing the optimal target influence \influencetarg. 
Lemmas \ref{lm:poly_hetero_binary} and \ref{lemma:ellip_closed_form} characterize these solutions for polyhedral and uncorrelated ellipsoidal uncertainty respectively.

\begin{lemma}[Robust Solution for Polyhedral Uncertainty]
\label{lm:poly_hetero_binary}
For the polyhedral uncertainty set $\mathcal{U}_s^{\text{poly}}$ with local bounds $\polybounds \in [0, 1]^n$ and noise budget $\gamma \in [0, n]$, the optimal influence $\influencetarg$ falls into two cases:
\squishlist
    \item \textbf{Case 1 (Full-Reliance):} All influence is allocated to a single node $\ell$, yielding an influence vector $\influencetarg$, with $q_{\ell}^* = 1$ and $q_{j \neq \ell}^* = 0$. The resulting public opinion is:
    \begin{equation}
        \publicnoisy_{\ell}(\influencetarg) = s_{\ell}(1 - \rho_{\ell}).
    \end{equation}
    \item \textbf{Case 2 (Equalization):} Influence is distributed across the top-$k$ nodes with the highest input inner opinions $s_{i}$, such that their contribution  $c_i = q_i s_{i}$  is equalized. The influence vector $\influencetarg$ satisfies:
    \begin{align}
        q_i^* = 
        \begin{cases} 
            \frac{1}{s_{i}} \left( \sum_{j=1}^k \frac{1}{s_{j}} \right)^{-1}, & \text{if } i \le k \\
            0, & \text{if } i > k.
        \end{cases}
    \label{eq:closed_poly} 
    \end{align}
    The resulting public opinion achieved by this distribution is:
        \begin{equation}
        \publicnoisy_{k}(\influencetarg) = \frac{k - \gamma}{\sum_{j=1}^k \frac{1}{s_{j}}}.
        \end{equation}
\squishend

The optimal influence vector can be computed as:
\begin{equation}
    \influencetarg = \arg \max \left( \max_{\ell \in \{1,\dots,n\}} \publicnoisy_{\ell}(\influencetarg), \ \max_{k \in \{\lfloor\gamma\rfloor+1,\dots,n\}} \publicnoisy_{k}(\influencetarg) \right).
\end{equation}
\end{lemma}


Lemma \ref{lm:poly_hetero_binary} 
suggests the following algorithm
for computing $\influencetarg$. 
First, we sort the input inner opinions $\inner$ in descending order. 
Next, we compute $\publicnoisy_{\ell}$ for all nodes and $\publicnoisy_{k}$ for $k \in \{\lfloor\gamma\rfloor+1,\dots,n\}$.
Finally, we extract the global maximum across these cases to find the optimal strategy. 
We call this algorithm {\CustomLemmaThree} and we provide full pseudocode in Alg.~\ref{alg:poly_solver}. The algorithm has
$\mathcal{O}(n \log n)$ running time.

\begin{algorithm}[ht]
\caption{\CustomLemmaThree}
\label{alg:poly_solver}
\begin{algorithmic}[1]
\State \textbf{Input:} Input inner opinions $\mathbf{s} \in (0,1]^n$, local bounds $\boldsymbol{\rho} \in [0,1]^n$, noise budget $\gamma \ge 0$
\State \textbf{Output:} Optimal target influence $\mathbf{q}^*$
\State Sort nodes such that $s_1 \ge s_2 \ge \dots \ge s_n$
\State Initialize $\publicnoisy_{\text{max}} \gets -\infty$, $\mathbf{q}^* \gets \mathbf{0}$
\State Initialize optimal strategy trackers: $\text{is\_equalization} \gets \text{False}$, $\ell^* \gets 0$, $k^* \gets 0$, $S^* \gets 0$

\State {\textbf{Case 1: Evaluate Full-Reliance Strategy}}
\For{$\ell = 1, \dots, n$}
    \State $\publicnoisy \gets s_\ell (1 - \rho_\ell)$
    \If{$\publicnoisy > \publicnoisy_{\text{max}}$}
        \State $\publicnoisy_{\text{max}} \gets \publicnoisy$
        \State $\ell^* \gets \ell$
        \State $\text{is\_equalization} \gets \text{False}$
    \EndIf
\EndFor

\State {\textbf{Case 2: Evaluate Equalization Strategy}}
\State Initialize running sum $S_{\text{inv}} \gets \sum_{j=1}^{\lfloor\gamma\rfloor} \frac{1}{s_j}$
\For{$k = \lfloor\gamma\rfloor + 1, \dots, n$}
    \State $S_{\text{inv}} \gets S_{\text{inv}} + \frac{1}{s_k}$
    \State $\publicnoisy \gets \frac{k - \gamma}{S_{\text{inv}}}$
    \If{$\publicnoisy > \publicnoisy_{\text{max}}$}
        \State $\publicnoisy_{\text{max}} \gets \publicnoisy$
        \State $k^* \gets k$, \quad $S^* \gets S_{\text{inv}}$
        \State $\text{is\_equalization} \gets \text{True}$
    \EndIf
\EndFor

\State {\textbf{Find Optimal Target Influence Vector}}
\If{$\text{is\_equalization}$}
    \For{$i = 1, \dots, k^*$}
        \State $q_i^* \gets \frac{1}{s_i \cdot S^*}$
    \EndFor
\Else
    \State $q_{\ell^*}^* \gets 1$
\EndIf

\State \Return $\mathbf{q}^*$
\end{algorithmic}
\end{algorithm}

This result 
gives a clear characterization of the optimal target influence $\influencetarg$ under polyhedral uncertainty. 
\emph{Case 1} 
mirrors the nominal solution, by placing all influence ($q_{\ell}^* = 1$) on a single node. 
This is optimal only when a node 
is exceptionally robust with  a low local bound $\rho_{\ell}$, and  
possesses a high input inner opinion $s_{\ell}$. 
If no such node exists, the optimal choice shifts to \emph{Case 2}. 
Here, we distribute influence among the top-$k$ nodes with the highest $s_{i}$ values, such that the \textit{contributions} $c_i = q_i^* s_{i}$ of these nodes to the public opinion are equalized. 
In the worst case, $\gamma$ of this cumulative contribution is eliminated, yielding the optimal public opinion.

\begin{lemma}[Robust Solution for Uncorrelated Ellipsoidal Uncertainty]
\label{lemma:ellip_closed_form}
For the uncorrelated ellipsoidal uncertainty set $\mathcal{U}_s^{\text{ellip}}$ with a diagonal matrix ${\Sigma}$ and noise budget $\omega$, the optimal target influence $\influencetarg$ is defined by distributing influence to the top-$k$ nodes with the highest input inner opinions $s_{i}$ such that:
\begin{equation}
    q_i^* = 
    \begin{cases}
        \frac{s_{i} - \lambda_k}{s_{i}^2 \sigma_i^2}  \left(\sum_{j=1}^k \frac{s_{j} - \lambda_k}{s_{j}^2 \sigma_j^2} \right)^{-1} & \text{if } i \le k \\
        0 & \text{if } i > k
    \end{cases}
\label{eq:closed_ellip}    
\end{equation}
where: 
\begin{equation}
\label{eq:lambda_root}
    \lambda_k = \frac{B_k - \sqrt{B_k^2 - A_k C_k}}{A_k}
\end{equation}
and $A_k = \sum_{i=1}^k \frac{1}{(s_{i} \sigma_i)^2}$, $B_k = \sum_{i=1}^k \frac{1}{s_{i} \sigma_i^2}$, and $C_k = \sum_{i=1}^k \frac{1}{\sigma_i^2} - \omega^2$. 

\noindent The optimal $k$ and $\lambda_k$ satisfy:
\begin{equation}
\label{eq:ellip_cutoff}
    s_{k} > \lambda_k \ge s_{k+1}.
\end{equation}
\end{lemma}

Lemma \ref{lemma:ellip_closed_form} 
suggests the following algorithm
for computing the optimal influence $\influencetarg$. 
First, we sort the input inner opinions $\inner$ in descending order. 
Next, we compute the constants $A_k$, $B_k$, $C_k$ incrementally as running sums and evaluate $\lambda_k$ from Eq.~\eqref{eq:lambda_root}, while checking for the termination condition from Eq.~\eqref{eq:ellip_cutoff} in a single $\mathcal{O}(n)$ sweep. This results in $\bigO(n \log n)$ running time. 
We call this algorithm {\CustomLemmaFour} and we provide full pseudocode in Alg.~\ref{alg:ellip_solver}. 

\begin{algorithm}[ht]
\caption{\CustomLemmaFour}
\label{alg:ellip_solver}
\begin{algorithmic}[1]
\State \textbf{Input:} Input inner opinions $\mathbf{s} \in (0,1]^n$, variances $\boldsymbol{\sigma}^2 \in \mathbb{R}_{>0}^n$, noise budget $\omega \ge 0$
\State \textbf{Output:} Optimal target influence $\mathbf{q}^*$
\State Sort nodes such that $s_1 \ge s_2 \ge \dots \ge s_n$

\If{$\omega = 0$}
    \State $\mathbf{q}^* \gets \mathbf{0}$, \quad $q_1^* \gets 1$ 
    \State \Return $\mathbf{q}^*$
\EndIf

\State Define boundary condition $s_{n+1} \gets -\infty$
\State Initialize running sums $A \gets 0$, $B \gets 0$, $C \gets -\omega^2$

\For{$k = 1, \dots, n$}
    \State $A \gets A + \frac{1}{(s_k \sigma_k)^2}$
    \State $B \gets B + \frac{1}{s_k \sigma_k^2}$
    \State $C \gets C + \frac{1}{\sigma_k^2}$
    
    \State $\lambda \gets \frac{B - \sqrt{B^2 - AC}}{A}$
    
    \If{$s_k > \lambda \ge s_{k+1}$}
        \State $\mathbf{q}^* \gets \mathbf{0}$
        \State $N_{\text{factor}} \gets B - \lambda A$ 
        \For{$i = 1, \dots, k$}
            \State $q_i^* \gets \frac{s_i - \lambda}{s_i^2 \sigma_i^2} \cdot \frac{1}{N_{\text{factor}}}$
        \EndFor
        \State \Return $\mathbf{q}^*$
    \EndIf
\EndFor
\end{algorithmic}
\end{algorithm}

This result from Lemma~\ref{lemma:ellip_closed_form} is similar to the Equalization strategy from Lemma \ref{lm:poly_hetero_binary}:
the optimal influence is distributed among the top-$k$ nodes with the highest $s_{i}$ values.
However, instead of equalizing the contributions $c_i$, ellipsoidal uncertainty forces the equalization of the terms: ${q_i^* (s_{i} \sigma_i)^2} / ({s_{i} - \lambda_k})$. 
Now, a node receives influence proportional to how far its input inner opinion $s_{i}$ exceeds the threshold $\lambda_k$, penalized by its local uncertainty $\sigma_i^2$.

Detailed proofs for Lemmas \ref{lm:poly_hetero_binary}, \ref{lemma:ellip_closed_form} are provided in Appendix \ref{app:Closed Form Algs}. 
For correlated ellipsoidal uncertainty, where matrix $\Sigma$ is non-diagonal, the optimal target influence $\influencetarg$ is obtained by directly solving the SOCP.
Modern interior-point conic solvers efficiently handle large-scale sparse SOCPs, making this step computationally tractable even for large real-world networks.
By applying these methods, we efficiently solve Step 1 and find the target influence $\influencetarg$.

\subsection{Step 2: Achieving the Target Influence 
{\influencetarg}}

Having found the target influence $\influencetarg$ in Step 1, the remaining challenge is to find the stubbornness vector $\stubbotarg \in [0,1]^n$ that achieves it. 
Because the influence distribution $\mathbf{q}(\mathbf{a})$ is a complex, non-convex function of the stubbornness vector {\stubbo}, as shown by \cite{ChanCorrection}, solving {Step 2} via standard optimization is computationally hard.
Here, we establish a novel theoretical result that finds $\stubbotarg$ in linear time.

\begin{theorem}[Target Influence Realization]
\label{thm:exact_realization}
Given a strongly connected graph $G$ with a row-stochastic interaction matrix $W$, for any target influence distribution $\influencetarg$, there always exists a stubbornness vector $\stubbotarg \in [0,1]^n$ such that $\influence(\stubbotarg) = \influencetarg$. The stubbornness values $a_i^*$ can be constructed as follows:
\begin{equation}
\label{eq:optimal_stubbornness}
    a_i^* = 
        \frac{q_i^*}{y_{0,i} + \eta \pi_i},
\end{equation}
where
\squishlist
    \item $\station$ is the stationary distribution of the random walk on graph $G$, with transition matrix $W$.
    \item $\mathbf{y}_0$ is any base solution vector to the sparse linear system: 
    \begin{equation}
    (I - W^T)\mathbf{y} = \frac{1}{n}\mathbf{1} - W^T\influencetarg,
    \end{equation}
    \item $\eta$ is a scalar shift required to bound the stubbornness values $a_i^*$ within $[0,1]$, with:
    \begin{equation}
        \eta \geq \eta_{\text{min}} = \max_{i : q_i^* > 0} \left( \frac{q_i^* - y_{0,i}}{\pi_i} \right).
    \end{equation}
\squishend
\end{theorem}

A detailed proof for Theorem \ref{thm:exact_realization} is provided in Appendix \ref{app:Feasibility}. 
Crucially, this theorem provides a scalable algorithm to compute the stubbornness $\stubbotarg$ required to achieve any target influence $\influencetarg$. 
First, computing the stationary distribution $\station$ of the interaction matrix $W$ via power iterations requires $\mathcal{O}(|E|)$ time. 
Next, finding vector $\mathbf{y}_0$ involves solving the sparse linear system $(I - W^T)\mathbf{y} = \frac{1}{n}\mathbf{1} - W^T\influencetarg$. 
Using the LSMR algorithm \cite{fong2011lsmr} (see Appendix \ref{app:LSMR}), this step also requires $\mathcal{O}(|E|)$ time per iteration. 
Finally, determining the range of $\eta$ by calculating $\eta_{\text{min}}$ requires $\mathcal{O}(n)$ time. 
Consequently, the execution of Step 2 is dominated by sparse matrix-vector multiplications, achieving a linear time complexity with respect to the number of nodes and edges.

\spara{Computational Complexity:} 
Summarizing,  
in Step 1, finding {\influencetarg} for polyhedral and uncorrelated ellipsoidal uncertainty requires $\mathcal{O}(n \log n)$ time 
(Lemmas~\ref{lm:poly_hetero_binary} and \ref{lemma:ellip_closed_form}). 
For correlated ellipsoidal uncertainty, we solve the SOCP (Lemma \ref{lemma:reform}).
In the worst case this  requires
 $\mathcal{O}(n^3)$ time \cite{Boyd2004}. However, in practice the actual running time is much faster due to the sparsity of matrix $\Sigma$ 
(see Sec.~\ref{sec:experiments}). 
Step 2 has ${\mathcal{O}}(|E|)$ complexity.
Overall, the two-step solution for \ProblemShort scales efficiently, making it applicable for massive  networks.

\section{The Budgeted {\ProblemShort} Problem}
\label{sec:budget}
In the previous section, we gave efficient and theoretically-sound solutions to the \ProblemShort problem. 
However, the lack of constraints on the interventions we perform leads to solutions that have a specific structure, with the stubbornness of a percentage of nodes becoming zero. 
While this is a valid theoretical result, these interventions may be extreme in certain situations.
To address this, we build on our two-step approach to introduce the budgeted version of \ProblemShort, which we call \ProblemBudget.


\begin{problem}[{\ProblemBudget}]
\label{problem:budgeted_romsa}
Assume a directed graph $G = (V, E)$ with interaction matrix $W$, and a vector of input inner opinions $\inner \in [0,1]^n$. Let $\mathcal{U}_s$ be the uncertainty set of the inner opinions. 
Given an initial stubbornness vector $\stubboInit$ and the corresponding initial influence vector \influenceInit, along with an intervention budget $b \in [0,2]$, find a stubbornness vector $\stubbotarg$ such that:
\begin{align}
\label{eq:b_ROMSA}
    \stubbotarg = \arg\max_{\stubbo \in [0,1]^n} \min_{\innernoisy \in \mathcal{U}_s} \publicnoisy(\stubbo, \innernoisy) \quad \text{s.t.} \quad \|\influence(\stubbotarg) - \influenceInit\|_1 \le b \, .
\end{align}
\end{problem}
Note that the intervention budget takes values $b \in [0,2]$ because $\influence$ is a \emph{distribution}, and that for $b=2$, \ProblemBudget becomes \ProblemShort. 
We can handle the {\ProblemBudget} problem within our two-step framework. 
Specifically, we append the constraint $\|\influence - \influenceInit\|_1 \le b$ to the constraints of the LP or SOCP from Lemma \ref{lemma:reform}.
Thus, the results of Lemma~\ref{lemma:reform} still hold for {\ProblemBudget} and 
we can compute the 
optimal influence target {\influencetarg} by solving the corresponding LP or SOCP, with the same time complexity. Step 2 is again solved via Theorem~\ref{thm:exact_realization}, which guarantees that all influence targets $\influencetarg$ can be achieved via adjustments in stubbornness.

To further enforce the notion of the budget for \ProblemBudget, we select the scalar $\eta$ in a specific manner. 
Of the infinitely many valid $\stubbotarg$ vectors that can achieve \influencetarg, we choose the one that minimizes the $L_2$ squared distance from the initial state {\stubboInit}. That is, we compute: $\eta^* = \arg\min_{\eta \ge \eta_{\text{min}}} \|\stubbotarg(\eta) - \stubboInit \|_2^2$. Because this optimization domain is strictly one-dimensional, we can efficiently compute $\eta^*$ in $\mathcal{O}(n)$ time (see Appendix~\ref{app:ETA} for more details).


\section{Experiments}\label{sec:experiments}

In our experiments, we evaluate our solution on real-world networks to address the following research questions (RQ's): 

\begin{enumerate}[label=RQ\arabic*:]
    \item How robust is our solution to worst-case noise?
    \item How well does our proposed solution scale?
    \item How does our robust solution distribute influence?
\end{enumerate}

\spara{Datasets and Hardware:} Table~\ref{tab:datasets_main} summarizes the characteristics of the SNAP datasets \cite{snapnets} we use to construct the adjacency matrices $W$. To ensure stationary Friedkin-Johnsen dynamics, we restrict all networks to their Largest Strongly Connected Component (LSCC). We execute all experiments on a 16GB Apple M4 MacBook Air. Source code for all experiments is available online.\footnote{\url{https://anonymous.4open.science/r/Robust-Opinion-Maximization-26D3}}

\spara{Parameters and Solvers:} We draw input inner opinions $\inner$ and initial stubbornness values $\stubboInit$ uniformly at random from $U(0,1)$. To avoid extreme cases of uncertainty, we draw local noise bounds $\rho_i$ and deviations $\sigma_i$ from $U(0.3, 0.7)$. To ensure statistical significance, we evaluate 10 random initializations of these parameters, reporting the mean and one standard deviation. For Step 1, we solve the LP and SOCP using Gurobi Optimizer's interior-point method. For Step 2, we compute $\mathbf{y}_0$ via the LSMR algorithm (see Appendix~\ref{app:LSMR}).

\spara{Correlated Noise:} To simulate ellipsoidal noise correlations, we construct a block-diagonal covariance matrix $\Sigma_c$ by applying the METIS algorithm \cite{METIS}. 
We specify an average community target size of $15$ nodes, resulting in $\max(1, \lfloor n / 15 \rfloor)$ partitions for each dataset. 
Intra-community noise correlation is set to $\Sigma_c[i,j] = 0.4 \sigma_i \sigma_j$, while inter-community noise is $\Sigma_c[i,j] = 0$.
We denote the diagonal matrix of uncorrelated ellipsoidal noise as $\Sigma_u$.

\begin{figure*}[ht] 
    \centering
    \vspace{-0.2cm} 
    
    \includegraphics[width=0.5\textwidth]{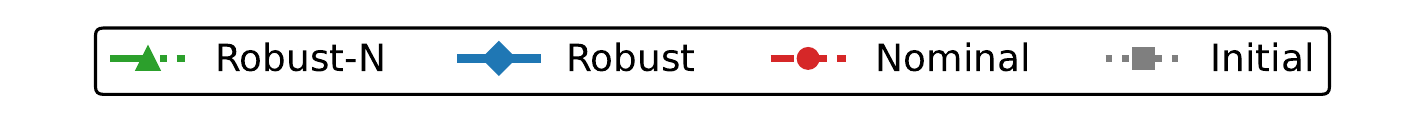}
    \vspace{-0.1cm} 

    \begin{subfigure}[b]{0.24\textwidth}
        \centering
        \includegraphics[width=\textwidth]{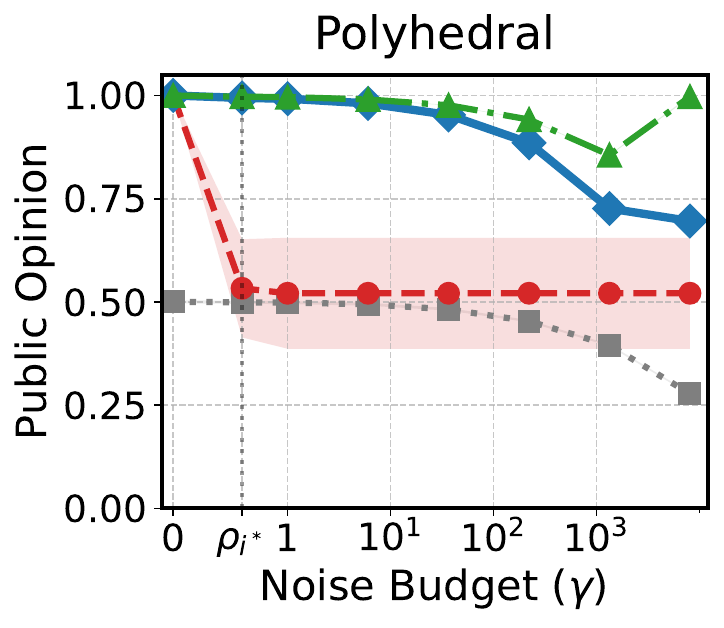}
        \vspace{-0.5cm} 
        \caption{}
        \label{fig:rq1_a}
    \end{subfigure}\hfill
    \begin{subfigure}[b]{0.24\textwidth}
        \centering
        \includegraphics[width=\textwidth]{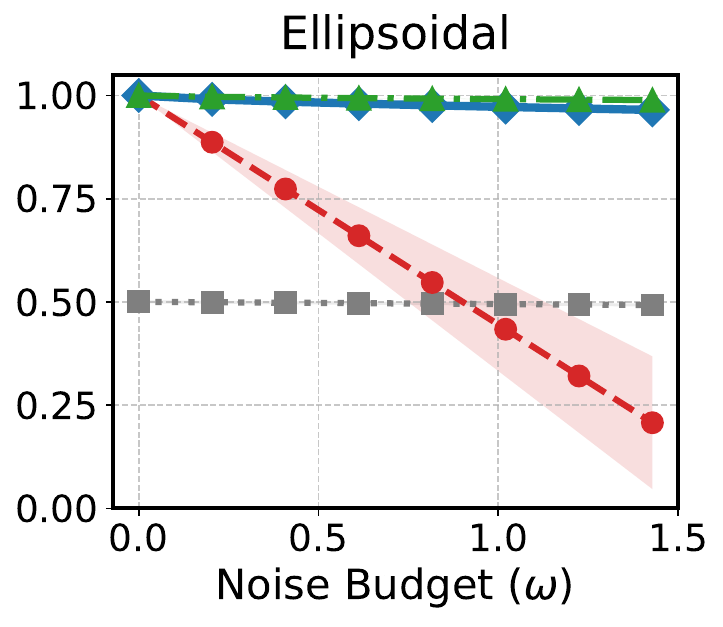}
        \vspace{-0.5cm} 
        \caption{}
        \label{fig:rq1_b}
    \end{subfigure}\hfill
    \begin{subfigure}[b]{0.24\textwidth}
        \centering
        \includegraphics[width=\textwidth]{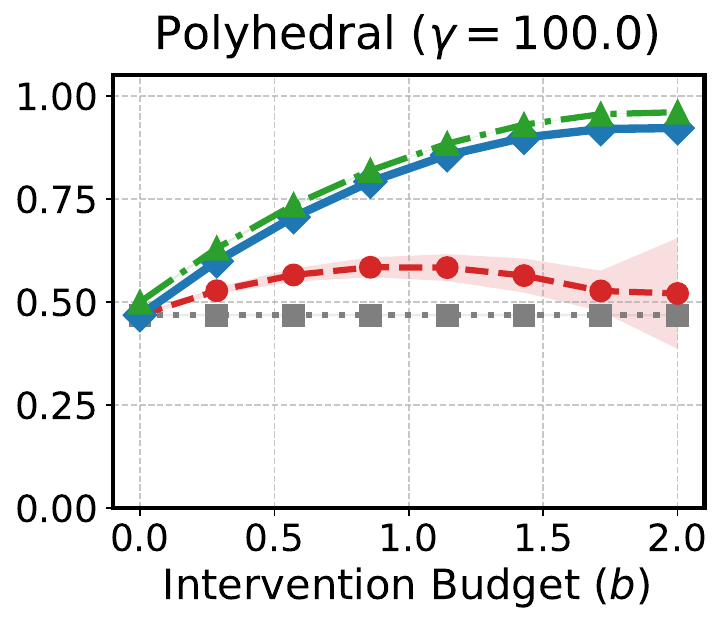}
        \vspace{-0.5cm} 
        \caption{}
        \label{fig:rq1_c}
    \end{subfigure}\hfill
    \begin{subfigure}[b]{0.24\textwidth}
        \centering
        \includegraphics[width=\textwidth]{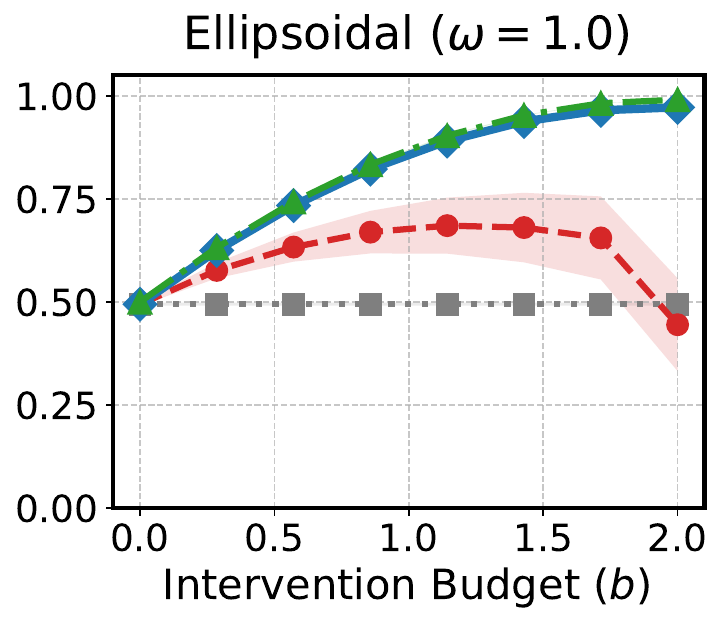}
        \vspace{-0.5cm} 
        \caption{}
        \label{fig:rq1_d}
    \end{subfigure}
    
    \vspace{-0.2cm} 
    \caption{Public opinion with varying (a), (b) noise budget and  (c), (d) intervention budget. }
    \label{fig:rq1_robustness}
    \vspace{-0.4cm} 
\end{figure*}

\begin{table}[ht]
    \centering
    \caption{Summary of network datasets.}
    \label{tab:datasets_main}
    \begin{tabular}{l r r l}
        \toprule
        \textbf{Dataset} & \textbf{Nodes ($n$)} & \textbf{Edges ($|E|$)} & \textbf{Type} \\
        \midrule
        Epinions & 32,223 & 443,506 & Directed \\
        Gowalla & 196,591 & 950,327 & Undirected \\
        Google & 434,818 & 3,419,124 & Directed \\
        Pokec & 1,304,537 & 29,183,655 & Directed \\
        \bottomrule
    \end{tabular}
\end{table}


\subsection{RQ1: How robust is our solution to worst-case noise?}

In this experiment, we evaluate the worst-case public opinion that our algorithms achieve in the presence of noise. 
We evaluate our solution (denoted \textsl{Robust}) against two baselines: \textsl{Nominal} maximizes public opinion assuming no noise ($\mathbf{\hat{s}}=\mathbf{s}$), while \textsl{Initial} applies no interventions to the initial stubbornness \stubboInit. We report experiments on the Epinions dataset. Results for all datasets are in Appendix \ref{app:Experiements_More}.


In Fig.~\ref{fig:rq1_robustness}(a) and Fig.~\ref{fig:rq1_robustness}(b), we show the worst-case public opinion (y-axis) for the {\ProblemShort} problem (no influence budget) when varying noise budgets $\gamma$, and $\omega$ (x-axis). 
Our \textsl{Robust} solution (blue) consistently dominates both the \textsl{Nominal} (red) and \textsl{Initial} (grey) baselines. 
As noise increases, \textsl{Nominal} degrades rapidly,
since it greedily concentrates all influence on the node $i^*$ with the highest inner opinion $s_{i^*}$.
For noise budget $\gamma > \rho_{i^*}$, the public opinion of \textsl{Nominal} becomes $\rho_{i^*}s_{i^*}$, which is occasionally worse than \textsl{Initial}.
Ellipsoidal uncertainty exhibits 
similar trends for all approaches. 

To understand the effect of the worst-case scenario on our solution, we also plot \textsl{Robust-N}, which evaluates our robust solution under zero-noise.
The small gap between \emph{Robust} (blue) and \emph{Robust-N} (green) highlights 
that noise can only marginally reduce the public opinion for our robust solution. 
Also, note that \textsl{Robust-N} is close to the theoretical optimal of 1, indicating that accounting for noise does not severely affect the achieved public opinion.

In Fig.~\ref{fig:rq1_robustness}(c) and Fig.\ref{fig:rq1_robustness}(d), we fix the noise ($\gamma=100$, $\omega=1$) and vary the intervention budget $b$ for the {\ProblemBudget} problem. 
As $b$ increases, the \textsl{Robust} public opinion increases, effectively utilizing the extra budget. 
\textsl{Nominal}, initially increases but eventually degrades, since it allocates the additional budget to a small number of nodes, making it highly susceptible to worst-case attacks. 
Note that \textsl{Initial} is consistently worst, indicating 
that doing nothing leaves the network in a sub-optimal state. 

\subsection{RQ2: How well does our proposed solution scale?}

To address RQ2, Table~\ref{tab:total_runtime} reports the execution time of our approach across networks of increasing size. Notably, the runtimes remain on the order of seconds even for the largest graphs, demonstrating the strong scalability of our framework.

In the top section of Table~\ref{tab:total_runtime}, we isolate Step 1 for \ProblemShort to demonstrate the algorithmic speedup achieved by \CustomLemmaThree and \CustomLemmaFour.
We observe that our algorithms yield a speedup of over two orders of magnitude over the state-of-the-art Gurobi optimizer for solving the LP and SOCP in Step 1.

In the bottom section of Table~\ref{tab:total_runtime}, we evaluate the runtime of our two-step solution for the {\ProblemBudget} problem with an intervention budget of $b=1$ and noise bounds set to $\gamma=100$ and $\omega=1$. 
We observe that our solution still scales exceptionally well, handling the 1.3-million-node {Pokec} graph in just minutes. We report the standard deviations of the running times in Appendix~\ref{app:Experiements_More}.

\begin{table}[b] 
    \centering
    \small 
    
    \caption{Runtime (in seconds). Top: Solving Step 1 of {\ProblemShort}. Bottom: Solving the full {\ProblemBudget} problem ($b=1$).}
    \label{tab:total_runtime}
    
    \begin{tabular*}{\columnwidth}{@{\extracolsep{\fill}} l r r r r @{}} 
        \toprule
        \multirow{3}{*}{\textbf{Dataset}} & \multicolumn{4}{c}{\textbf{{\ProblemShort} Step 1}} \\
        \cmidrule(lr){2-5}
        & \multicolumn{2}{c}{\textbf{Polyhedral}} & \multicolumn{2}{c}{\textbf{Ellipsoidal}} \\
        \cmidrule(lr){2-3} \cmidrule(lr){4-5}
        & Solver & \CustomLemmaThree & Solver & \CustomLemmaFour \\
        \midrule
        Epinions & 0.27 & \textbf{0.0020} & 0.46 & \textbf{0.0021} \\
        Gowalla  & 1.68 & \textbf{0.0136} & 2.92 & \textbf{0.0148} \\
        Google   & 3.92 & \textbf{0.0332} & 6.77 & \textbf{0.0354} \\
        Pokec    & 14.01 & \textbf{0.1117} & 20.01 & \textbf{0.1168} \\
        \bottomrule
    \end{tabular*}
    
    \vspace{0.4cm} 
    
    \begin{tabular*}{\columnwidth}{@{\extracolsep{\fill}} l r r r r r r @{}} 
        \toprule
        \multirow{3}{*}{\textbf{Dataset}} & \multicolumn{6}{c}{\textbf{{\ProblemBudget} ($b=1$)}} \\
        \cmidrule(lr){2-7}
        & \multicolumn{3}{c}{\textbf{Polyhedral}} & \multicolumn{3}{c}{\textbf{Ellipsoidal}} \\
        \cmidrule(lr){2-4} \cmidrule(lr){5-7}
        & Step 1 & Step 2 & \textbf{Total} & Step 1 & Step 2 & \textbf{Total} \\
        \midrule
        Epinions & 1.59 & 4.98 & \textbf{6.57} & 2.70 & 3.78 & \textbf{6.48} \\
        Gowalla  & 17.23 & 33.47 & \textbf{50.70} & 23.01 & 19.65 & \textbf{42.66} \\
        Google   & 77.88 & 130.90 & \textbf{208.78} & 100.57 & 99.65 & \textbf{200.22} \\
        Pokec    & 136.91 & 278.72 & \textbf{415.63} & 263.50 & 229.33 & \textbf{492.83} \\
        \bottomrule
    \end{tabular*}
\end{table}

\begin{figure*}[t]
    \centering
    
    \captionsetup[subfigure]{margin={-1.2cm, 0cm}} 
    
    \begin{subfigure}[b]{0.32\textwidth}
        \centering
        \includegraphics[width=\textwidth]{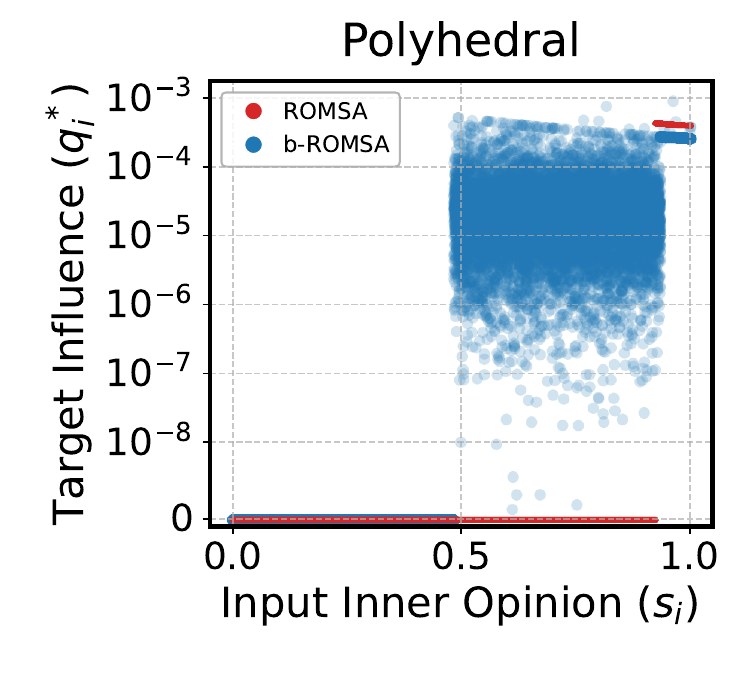}
        \vspace{-0.75cm} 
        \caption{}
        \label{fig:rq3_a}
    \end{subfigure}\hfill
    \begin{subfigure}[b]{0.32\textwidth}
        \centering
        \includegraphics[width=\textwidth]{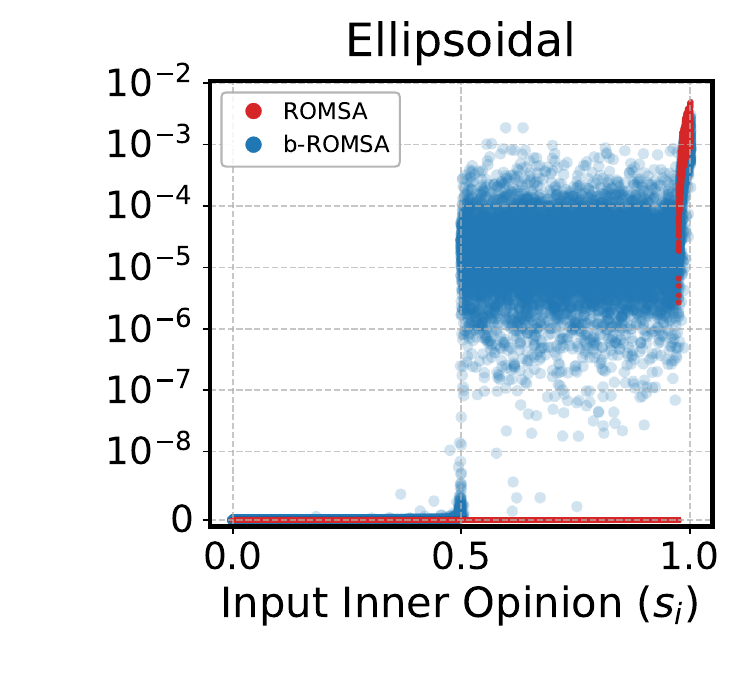}
        \vspace{-0.75cm} 
        \caption{}
        \label{fig:rq3_b}
    \end{subfigure}\hfill
    \begin{subfigure}[b]{0.32\textwidth}
        \centering
        \includegraphics[width=\textwidth]{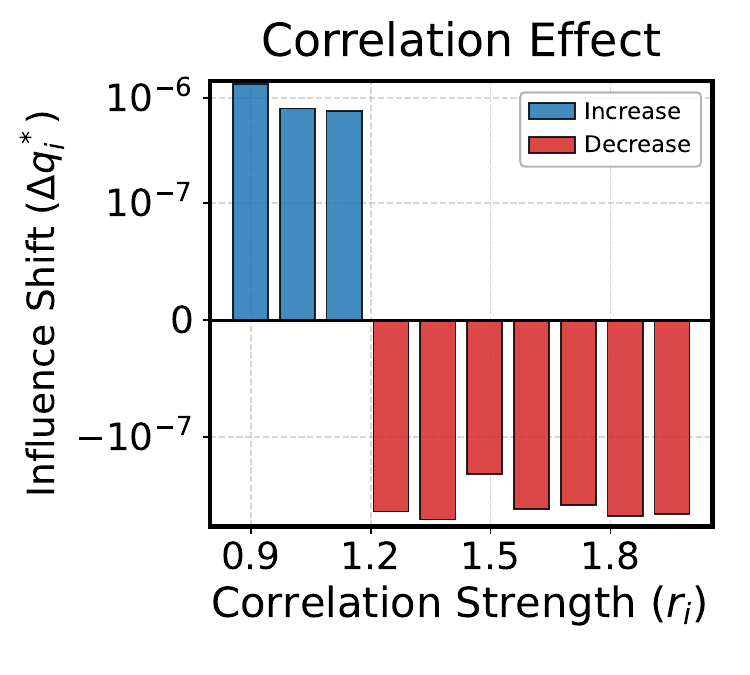} 
        \vspace{-0.75cm} 
        \caption{}
        \label{fig:rq3_c}
    \end{subfigure}
    
    \vspace{-0.2cm} 
    \caption{Scatter plot of $q_i^*,s_i$ for (a)  Polyhedral ($\gamma$ = 100) and (b) Ellipsoidal ($\omega$ = 1)  uncertainty. 
    (c) Influence shift ($\Delta q_i^*$) against correlation strength. All plots are on the Epinions dataset.}
    \label{fig:rq3_distributions}
    \vspace{-0.3cm} 
\end{figure*}

\begin{figure*}[b]
    \centering
    
    \begin{subfigure}[b]{0.24\textwidth}
        \centering
        \includegraphics[width=\textwidth]{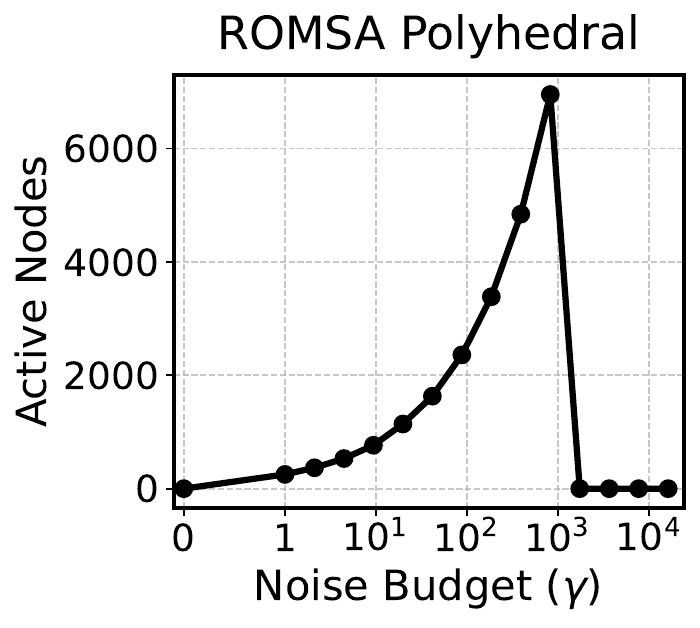}
        \caption{}
        \label{fig:app1_a}
    \end{subfigure}\hfill
    \begin{subfigure}[b]{0.24\textwidth}
        \centering
        \includegraphics[width=\textwidth]{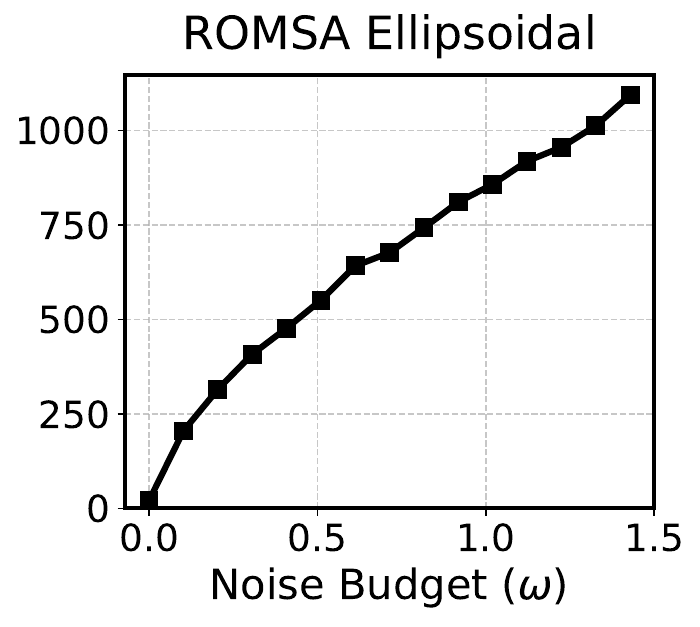}
        \caption{}
        \label{fig:app1_b}
    \end{subfigure}\hfill
    \begin{subfigure}[b]{0.24\textwidth}
        \centering
        \includegraphics[width=\textwidth]{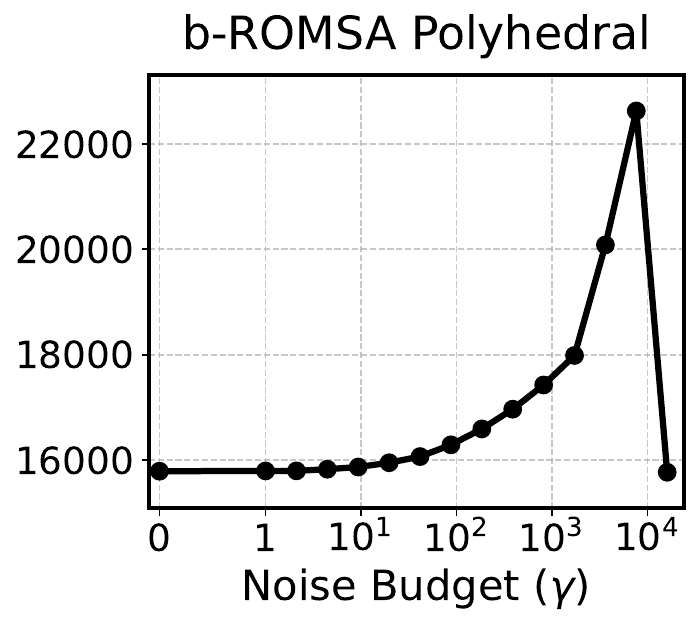}
        \caption{}
        \label{fig:app1_c}
    \end{subfigure}\hfill
    \begin{subfigure}[b]{0.24\textwidth}
        \centering
        \includegraphics[width=\textwidth]{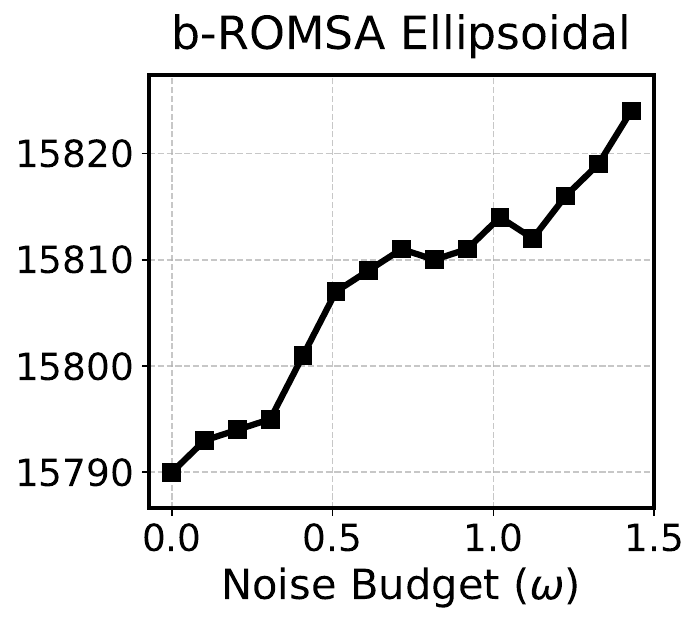}
        \caption{}
        \label{fig:app1_d}
    \end{subfigure}
    
    \caption{Number of active nodes as the noise budget increases. Plots (a), (b) show the results for \ProblemShort, while plots (c), (d) depict the results for \ProblemBudget with an intervention budget $b=1$. All plots are on the Epinions dataset.}
    \label{fig:rq3_diversification_combined}
\end{figure*}

\subsection{RQ3: How does our robust solution distribute influence?}

We now examine how our solution allocates $\influencetarg$. We analyze the structure of the resulting influence distributions, evaluate their sensitivity to correlated noise, and investigate how the number of active nodes ($q_i^* > 0$) shifts under increasing noise budgets.
%

In Fig.~\ref{fig:rq3_distributions}(a) and Fig.\ref{fig:rq3_distributions}(b) we have a scatter plot of the optimal influence values $q_i^*$ against the inner opinions $s_i$ for \ProblemShort (red) and \ProblemBudget (blue). 
For \ProblemShort, the algorithm distributes influence 
to the top $k$ nodes with the highest $s_i$, as determined by Eq.~(\ref{eq:closed_poly}) and Eq~(\ref{eq:closed_ellip}). 
For \ProblemBudget ($b=1$), the solver can no longer concentrate all influence on these top $k$ nodes.
Instead it is forced to distribute influence across a much broader set. 
However, despite this wider allocation, the solver still prioritizes the top-$k$ nodes, 
assigning them high influence $q_i^*$ as in the red distribution.

We then study the effect of the correlation matrix $\Sigma$ in ellipsoidal uncertainty on $\influencetarg$.
We compute  $\influencetarg_u$ with a diagonal correlation matrix $\Sigma_u$ (no correlation), and $\influencetarg_c$, with a  correlation matrix $\Sigma_c$ where nodes in the same community are correlated. We compute the influence shift $\Delta \influencetarg = \influencetarg_c - \influencetarg_u$, which is the effect of correlation.

In Fig.~\ref{fig:rq3_distributions}(c) we plot the influence shift ($\Delta q_i^*$) against the \emph{correlation strength} $r_i = \sum_{j \neq i} \Sigma_c[i,j]$, which measures how much the uncertainty of a node correlates with that of it's peers.
We observe that the optimal solution recognizes highly correlated nodes as vulnerabilities. High noise on these nodes will ``spread'' to 
the entire correlated community. 
To minimize this risk, the optimal solution removes influence from highly correlated nodes (red) and allocates it to more independent, low-correlation nodes (blue). 

We also analyze how the number of \emph{active nodes} ($q_i^* > 0$) for {\influencetarg} changes as the noise budgets $\gamma$ for polyhedral and $\omega$ for ellipsoidal uncertainty increase. We show this in Figure~\ref{fig:rq3_diversification_combined} for \ProblemBudget ($b=1$) vs. \ProblemShort under both polyhedral and ellipsoidal uncertainty.

For polyhedral uncertainty in Figures~\ref{fig:rq3_diversification_combined}a and \ref{fig:rq3_diversification_combined}c, we depict the network's \emph{active nodes} ($q_i^* > 0$).
Notice that \ProblemBudget begins with and maintains a larger number of active nodes compared to \ProblemShort. 
The intervention budget $b$ directly causes this by preventing influence from shifting too far away from $\influenceInit$. 
Despite this, both cases exhibit a distinct "peak-and-drop" pattern. 
As the noise budget $\gamma$ grows, it is optimal to diversify influence across many nodes. 
However, once $\gamma$ crosses a critical threshold, the noise becomes too powerful and diversification fails. 
The optimal solution changes strategy and executes \emph{full-reliance} by distributing all available influence onto the single safest node to control the worst-case damage. 
Interestingly, this collapse occurs in both \ProblemShort and \ProblemBudget.

For ellipsoidal uncertainty in Figures~\ref{fig:rq3_diversification_combined}b and \ref{fig:rq3_diversification_combined}d, we depict the network's \emph{active nodes} ($q_i^* > 0$). 
Here we see a smooth, continuous upward curve as the noise budget $\omega$ increases. 
This reflects the built-in diversification effect of the SOCP solution: to minimize the penalty of a stronger noise, it must spread influence across a wider range of nodes.
However, in \ProblemBudget, the budget constraint $b$ starkly limits the extent of this. 
While the \ProblemShort problem actively scales its active nodes from zero, the \ProblemBudget problem is anchored initially high (starting at 15,790 active nodes) and only activates a few dozen additional nodes (peaking at 15,825) as $\omega$ increases.

\section{Conclusion}
\label{sec:conclusion}

In this paper, we introduced the first robust optimization framework for opinion-formation dynamics, addressing the critical challenge of handling uncertainty in social networks. By leveraging classic robust optimization techniques, we formulated and solved the {\ProblemShort} and {\ProblemBudget} problems under the Friedkin-Johnsen model using a highly scalable, two-step approach capable of processing million-node graphs in minutes.

Our work establishes a foundation that opens new research directions for robust network interventions. While we focused on mitigating uncertain inner opinions via stubbornness adjustments, these parameters can be altered to facilitate different strategies, such as modifying graph topology against uncertainty in user stubbornness. Alternative objectives, like minimizing polarization, can also be introduced. Future work can explore different intervention budgets and constraint structures. Finally, while we achieved optimal $\mathcal{O}(n \log n)$ speedups for the unbudgeted cases, deriving comparable algorithmic speedups for correlated ellipsoidal noise and under budget constraints remains an open challenge.

\bibliographystyle{ACM-Reference-Format}
\bibliography{robustness-bibliography}


\appendix

\section{Constructing a Boundary-Respecting Ellipsoidal Uncertainty Set}
\label{appendix:covariance_setup}

In this section, we derive the condition required to guarantee that worst-case noise {\noise} results in valid noisy opinions {\innernoisy}.
Recall the definition of the ellipsoidal uncertainty set from Equation \eqref{noise_elips}:
\begin{equation*}
    \mathcal{U}_s^{\text{ellip}} = \left\{ \mathbf{\hat{s}} = \mathbf{s} \circ (\mathbf{1} - \boldsymbol{\delta}) \; \middle| \; \boldsymbol{\delta}^T \Sigma^{-1} \boldsymbol{\delta} \le \omega^2 \right\}.
\end{equation*}

Since we are solving a robust maximization problem over the FJ noisy public opinion $\publicnoisy = {\influence}^T {\innernoisy}$, the worst-case noise corresponds to the maximum possible reduction of the input inner opinions {$s_i$}. 
Because the input inner opinions $s_{i}$ are non-negative, the maximum valid reduction occurs when an inner opinion is reduced to zero, which corresponds to a noise of $\delta_i = 1$. 
Thus, we must enforce an upper bound of $\delta_i \le 1$ for all valid noises. 

Next, we evaluate the maximum noise possible under the defined ellipsoidal uncertainty set. We seek the maximum value that the $i$-th component $\delta_i$ can reach within the set $\mathcal{U}_s^{\text{ellip}}$. This is determined solely by the ellipsoidal constraint on $\boldsymbol{\delta}$:
\begin{equation*}
    \max_{\boldsymbol{\delta}} \; \delta_i \quad \text{subject to} \quad \boldsymbol{\delta}^T \Sigma^{-1} \boldsymbol{\delta} \le \omega^2.
\end{equation*}
This optimization problem corresponds to finding the maximum projection of the ellipsoid defined by $\Sigma$ and $\omega$ onto the $i$-th coordinate axis. Using standard results for the support function of an ellipsoid centered at the origin \cite{Boyd2004}, this maximum value evaluates to:
\begin{equation*}
    \max \delta_i = \omega \sqrt{\Sigma[{i,i}]}.
\end{equation*}

Crucially, this maximum projection depends \textit{only} on the diagonal element $\Sigma[{i,i}]$ and the budget $\omega$. The off-diagonal elements $\Sigma[{i,j}]$ (correlations) affect the spatial orientation of the ellipsoid but cannot extend its absolute projection along any single axis.

To ensure feasibility for all possible realizations of worst-case noise, we require that this maximum possible deviation does not exceed the limit of $1$:
\begin{equation*}
    \omega \sqrt{\Sigma[{i,i}]} \le 1.
\end{equation*}

Recalling our definition of the covariance diagonal entries as $\Sigma[{i,i}] = \sigma_i^2$, we can rewrite this local condition as $\omega \sigma_i \le 1$. To guarantee that this condition holds simultaneously for all nodes, the global budget parameter $\omega$ must be constrained by the maximum noise variance present in the network. Consequently, we arrive at the strict global bound:
\begin{equation*}
    \omega \le \frac{1}{\max_i \sigma_i}.
\end{equation*}

By constructing $\Sigma$ and selecting $\omega$ within this range, we ensure that the maximum possible noise respects the lower bound ($\hat{s}_i \ge 0$). 

\spara{The Ellipsoidal Relaxation:} 
It is important to note that the standard ellipsoid $\boldsymbol{\delta}^T \Sigma^{-1} \boldsymbol{\delta} \le \omega^2$ is symmetric around the origin, meaning it permits negative noise values ($\delta_i < 0$). 
A negative noise value would increase an input inner opinion ($\hat{s}_i > s_i$). However, because the inner minimization strictly wants to reduce the objective, it will always choose positive noise values ($\delta_i \ge 0$). 
Therefore, allowing the uncertainty set to extend into the negative $\noise$ space acts as a harmless geometric relaxation. 
We utilize this relaxation because explicitly appending non-negativity box constraints ($\delta_i \in [0,1]$) to the definition of the ellipsoidal uncertainty set would create the intersection of an ellipsoid and a hypercube. 
Because this intersection is not a pure ellipsoid, it breaks the strict geometric properties required for standard robust optimization \cite{bental1998robust}. 
By enforcing our global bound on $\omega$ instead, we guarantee both physical validity and computational tractability.

\section{Solving the Inner Minimization Problem of ROMSA}
\label{appendix:inner_min}

In this section, we derive the solutions for the inner minimization problem when the stubbornness vector $\mathbf{a}$ (and thus the influence vector $\mathbf{q}(\mathbf{a})$) is fixed.

Given a fixed $\mathbf{a}$, we seek to minimize the \emph{noisy public opinion} $\publicnoisy$. Using the multiplicative noise model $\hat{s_i} = s_{i}(1 - \delta_i)$, the objective becomes:
\begin{equation*}
    \min_{\boldsymbol{\delta}} \sum_{i=1}^n q_i(\mathbf{a}) s_{i} (1 - \delta_i) = \sum_{i=1}^n q_i(\mathbf{a}) s_{i} - \max_{\boldsymbol{\delta}} \sum_{i=1}^n q_i(\mathbf{a}) s_{i} \delta_i
\end{equation*}
Let $c_i = q_i(\mathbf{a}) s_{i} \ge 0$ denote the \emph{contribution} of node $i$. The worst-case noise can be found by solving the maximization problem:
\begin{equation}
\label{eq:adv_objective}
    \max_{\boldsymbol{\delta}} \mathbf{c}^T \boldsymbol{\delta}
\end{equation}

\spara{Polyhedral Set Strategy (Continuous Knapsack):}
For the polyhedral set, the constraints on the noise is: $0 \le \delta_i \le \rho_i$ and $\sum_{i=1}^n \delta_i \le \gamma$. 
The problem in Eq.~\eqref{eq:adv_objective} is a standard continuous knapsack problem. The optimal strategy has the following greedy steps:
1. Sort the nodes in descending order of their contribution coefficients $c_i$.
2. Iteratively set $\delta_i = \rho_i$ for the highest-ranked nodes until the budget $\gamma$ is exhausted. 
3. The last affected node receives any fractional remainder of the budget, and all lower-ranked nodes receive $\delta_j = 0$.

\spara{Ellipsoidal Set Strategy (Gradient Alignment):}
For the ellipsoidal set, the noise is constrained by $\boldsymbol{\delta}^T {\Sigma}^{-1} \boldsymbol{\delta} \le \omega^2$. Note that no explicit box constraints $\delta_i\in[0,1]$ are needed because $\omega \in \left[0, \, 1 / \max_i \sigma_i \right]$ (see Appendix \ref{appendix:covariance_setup}).
Let $\mathbf{u} = {\Sigma}^{-1/2} \boldsymbol{\delta}$. The constraint becomes $\mathbf{u}^T \mathbf{u} = \|\mathbf{u}\|_2^2 \le \omega^2$, and the objective becomes $\max_{\mathbf{u}} ({\Sigma}^{1/2} \mathbf{c})^T \mathbf{u}$.
The maximum is achieved when $\mathbf{u}$ is perfectly aligned with ${\Sigma}^{1/2} \mathbf{c}$:
\begin{equation*}
    \mathbf{u}^* = \omega \frac{{\Sigma}^{1/2} \mathbf{c}}{\|{\Sigma}^{1/2} \mathbf{c}\|_2}.
\end{equation*}
Transforming back to $\boldsymbol{\delta} = {\Sigma}^{1/2} \mathbf{u}$, the worst-case noise is given in closed form by:
\begin{equation*}
    \boldsymbol{\delta} = \omega \frac{{\Sigma} \mathbf{c}}{\sqrt{\mathbf{c}^T {\Sigma} \mathbf{c}}}.
\end{equation*}
This demonstrates that for both sets, if the network stubbornness parameter $\mathbf{a}$ is fixed, the inner minimization problem is easy to solve. The true complexity of ROMSA arises from optimizing $\mathbf{q}(\mathbf{a})$ against a noise that "reacts" to our choices.

\section{Proof of Lemma \ref{lemma:reform}}
\label{appendix:robust_proofs}

In this section, we provide the full proofs and resulting LP and SOCP from equation (\ref{eq:transformation}) for Step 1 under both polyhedral and ellipsoidal uncertainty sets.
Recall that our goal is to solve:
\begin{equation*}
    \influencetarg = \arg\max_{\mathbf{q}} \min_{\mathbf{\hat{s}} \in \mathcal{U}_s} \mathbf{q}^T \mathbf{\hat{s}}
\end{equation*}
By substituting the multiplicative noise model $\hat{s_i} = s_{i}(1 - \delta_i)$, the inner minimization becomes:
\begin{equation}
\label{eq:inner_sub_form}
    \min_{\boldsymbol{\delta}} \sum_{i=1}^n q_i s_{i} (1 - \delta_i) = \mathbf{q}^T \mathbf{s} - \max_{\boldsymbol{\delta}} \sum_{i=1}^n q_i s_{i} \delta_i.
\end{equation}

\subsection{Polyhedral Uncertainty}
For the polyhedral set $\mathcal{U}_s^{poly}$, the constraints on the noise vector $\boldsymbol{\delta}$ are $0 \le \delta_i \le \rho_i$ and $\sum_{i=1}^n \delta_i \le \gamma$. The maximization problem from Eq.~\eqref{eq:inner_sub_form} is the following Linear Program:
\begin{equation*}
\begin{aligned}
    & \max_{\boldsymbol{\delta}} \quad && \sum_{i=1}^n q_i s_{i} \delta_i \\
    & \text{subject to} \quad && \sum_{i=1}^n \delta_i \le \gamma \\
    & && \delta_i \le \rho_i, \quad \forall i \in \{1, \dots, n\} \\
    & && \delta_i \ge 0, \quad \forall i \in \{1, \dots, n\}
\end{aligned}
\end{equation*}
To construct the dual of this LP, we associate a dual variable $\theta \ge 0$ with the budget constraint, and a dual variable $p_i \ge 0$ with each individual bound constraint $\delta_i \le \rho_i$. The dual problem is:
\begin{equation*}
\begin{aligned}
    & \min_{\theta, \boldsymbol{p}} \quad && \gamma \theta + \sum_{i=1}^n \rho_i p_i \\
    & \text{subject to} \quad && \theta + p_i \ge q_i s_{i}, \quad \forall i \in \{1, \dots, n\} \\
    & && \theta \ge 0, \quad p_i \ge 0, \quad \forall i \in \{1, \dots, n\}
\end{aligned}
\end{equation*}
By strong duality, the optimal values of the primal and dual are equal. 
We substitute the dual back into Eq.~\eqref{eq:inner_sub_form}. Since the original objective subtracts this term, the minimization becomes a maximization:
\begin{equation*}
    \max_{\mathbf{q}} \left( \mathbf{q}^T \mathbf{s} - \min_{\theta, \boldsymbol{p}} \left[ \gamma \theta + \sum_{i=1}^n \rho_i p_i \right] \right) = \max_{\mathbf{q}, \theta, \boldsymbol{p}} \left( \mathbf{q}^T \mathbf{s} - \gamma \theta - \sum_{i=1}^n \rho_i p_i \right).
\end{equation*}
Including the constraints from the dual problem and the constraints for $\mathbf{q}$ to be an influence distribution, yields the Linear Program (LP):
\begin{equation}
\label{eq:poly_robust}
\begin{aligned}
    & \max_{\influence, \theta, \boldsymbol{p}} \quad && \influence^T \inner - \gamma \theta - \sum_{i=1}^n \rho_i p_i \\
    & \text{\emph{subject to}} \quad && \theta + p_i \ge q_i s_{i}, \quad \forall i \in \{1, \dots, n\} \\
    & && \sum_{i=1}^n q_i = 1, \quad  q_i \ge 0, \quad \forall i \in \{1, \dots, n\} \\
    & && \theta \ge 0, \quad p_i \ge 0, \quad \forall i \in \{1, \dots, n\}
\end{aligned}
\end{equation} \hfill $\blacksquare$

\subsection{Ellipsoidal Uncertainty}
For the ellipsoidal set $\mathcal{U}_s^{\text{ellip}}$, the noise vector $\noise$ is bounded as $\boldsymbol{\delta}^T {\Sigma}^{-1} \boldsymbol{\delta} \le \omega^2$.
Defining $\mathbf{c}$ such that $c_i = q_i s_{i}$ (or $\mathbf{c} = \text{diag}(\mathbf{s})\mathbf{q}$ in matrix form, where $\text{{diag}}(\mathbf{s})$ is the diagonal matrix of the input inner opinions), the maximization problem from Eq.~\eqref{eq:inner_sub_form} is:
\begin{equation*}
    \max_{\boldsymbol{\delta}} \quad \mathbf{c}^T \boldsymbol{\delta} \quad \text{subject to} \quad \boldsymbol{\delta}^T {\Sigma}^{-1} \boldsymbol{\delta} \le \omega^2.
\end{equation*}
As shown in Appendix \ref{appendix:inner_min} the maximum is achieved when: 

$$\boldsymbol{\delta} = \omega \frac{{\Sigma} \mathbf{c}}{\sqrt{\mathbf{c}^T {\Sigma} \mathbf{c}}}$$ 

Thus the exact maximum value is given analytically by:
\begin{equation*}
    \omega \sqrt{\mathbf{c}^T {\Sigma} \mathbf{c}} = \omega \sqrt{ (\text{diag}(\mathbf{s})\mathbf{q})^T {\Sigma} (\text{diag}(\mathbf{s})\mathbf{q}) }.
\end{equation*}
Using the Cholesky decomposition ${\Sigma} = ({\Sigma}^{1/2})^T {\Sigma}^{1/2}$, we can rewrite this term as an L2-norm:
\begin{equation*}
    \omega \left\| {\Sigma}^{1/2} \text{diag}(\mathbf{s}) \mathbf{q} \right\|_2.
\end{equation*}
Substituting this closed-form penalty back into the outer problem, we have:
\begin{equation*}
    \max_{\mathbf{q}} \quad \mathbf{q}^T \mathbf{s} - \omega \left\| {\Sigma}^{1/2} \text{diag}(\mathbf{s}) \mathbf{q} \right\|_2.
\end{equation*}
To format this as a standard Second-Order Cone Program (SOCP), we introduce the variable $t \in \mathbb{R}$ to bound the norm:
\begin{equation*}
    \left\| {\Sigma}^{1/2} \text{diag}(\mathbf{s}) \mathbf{q} \right\|_2 \le t.
\end{equation*}
Replacing the norm in the objective function with $t$ yields the Second-Order Cone Program (SOCP):
\begin{equation}
\label{eq:ellip_robust}
\begin{aligned}
    & \max_{\influence, t} \quad && \influence^T \inner - \omega t \\
    & \text{\emph{subject to}} \quad && \left\| {\Sigma}^{1/2} \text{\emph{diag}}(\inner) \influence \right\|_2 \le t \\
    & && \sum_{i=1}^n q_i = 1, \quad  q_i \ge 0, \quad \forall i \in \{1, \dots, n\}
\end{aligned}
\end{equation}
\hfill $\blacksquare$

\section{Proofs for Lemmas \ref{lm:poly_hetero_binary} and \ref{lemma:ellip_closed_form}}
\label{app:Closed Form Algs}

In this appendix, we provide the detailed mathematical proofs and derivations for the results presented in Section \ref{sec:algorithm}.

\subsection{Proof of Lemma \ref{lm:poly_hetero_binary}: Polyhedral Noise}

\begin{proof}
To prove the optimal target influence $\mathbf{q}^*$ strictly falls into one of two categories, we analyze the LP formulation derived in Lemma \ref{lemma:reform}. Let $A$ be the active set of nodes where $q_i > 0$. From the robust formulation in Eq.~\eqref{eq:poly_robust}, the problem contains the inequality constraints $\theta + p_i \ge q_i s_{i}$ for all $i$. 
When the budget $\gamma > 0$ and local bounds $\rho_i > 0$, the variables $\theta$ and $p_i$ strictly penalize the objective function. 
By definition, the optimal solver will drive $p_i$ to its lowest valid limit: $p_i = \max(0, q_i s_i - \theta)$. 
We establish that in the optimal solution, no active node can remain strictly below the threshold ($q_i s_i < \theta$). 
If such a node existed, its penalty would be $p_i = 0$, and the solver could strictly improve the objective by reallocating influence between this node and other active nodes until it either reaches the threshold ($q_i s_i = \theta$) or is completely divested ($q_i = 0$), rendering it inactive. 
Consequently, for all nodes in the optimal active set, the constraint is perfectly tight: $q_i s_{i} = \theta + p_i$.

Using this tight constraint, we substitute the control variable $q_i = (\theta + p_i)/s_{i}$ into the distribution constraint $\sum_{i \in A} q_i = 1$, yielding a single constraint for the dual variables $\theta$ and $\mathbf{p}$:
\begin{equation}
\label{eq:knapsack_constraint}
    \theta \sum_{i \in A} \frac{1}{s_{i}} + \sum_{i \in A} \frac{p_i}{s_{i}} = 1.
\end{equation}
Next, substituting $q_i s_{i} = \theta + p_i$ into the objective function allows us to rewrite the maximization as:
\begin{equation}
\label{eq:knapsack_objective}
    \max_{\theta, p_i \ge 0} \quad \theta(|A| - \gamma) + \sum_{i \in A} p_i(1 - \rho_i).
\end{equation}
This transformed problem is exactly a continuous (fractional) knapsack problem with a single capacity constraint equal to 1. We can define the objective coefficients as "Values" ($V_\theta = |A| - \gamma$ and $V_i = 1 - \rho_i$) and the constraint coefficients as "Costs" ($C_\theta = \sum_{i \in A} 1/s_{i}$ and $C_i = 1/s_{i}$). By the fundamental theorem of linear programming, the optimal solution to a single-constraint continuous knapsack problem allocates the entire capacity strictly to the single variable possessing the highest efficiency ratio (Value/Cost). 

We evaluate these efficiency ratios, denoted as $R$:
\begin{itemize}
    \item For the variable $\theta$, the ratio is $R_\theta = \frac{V_\theta}{C_\theta} = \frac{|A| - \gamma}{\sum_{i \in A} 1/s_{i}}$, which is identical to $\publicnoisy_{k}(|A|)$.
    \item For any variable $p_i$, the ratio is $R_i = \frac{V_i}{C_i} = \frac{1 - \rho_i}{1/s_{i}} = s_{i}(1 - \rho_i)$, which is identical to $\publicnoisy_{\ell}(i)$.
\end{itemize}

Because the optimal solution selects exactly \emph{one} variable to be non-zero:
\begin{enumerate}
    \item If a specific $\ell \in \{1,\dots,n\}$ gives $R_{\ell}$ as the global maximum, the solver sets $p_{\ell} > 0$ and $\theta = 0$ (with $p_{j \neq \ell} = 0$). From the tight constraint, $q_{\ell} s_{\ell} = p_{\ell}$, and $q_j s_{j} = 0$. Applying the influence constraint $\sum q_i = 1$ forces $q_{\ell} = 1$. The solver allocates all influence to node $\ell$ (\textbf{Case 1: Full-Reliance}).
    \item If $R_\theta$ is the global maximum, the solver sets $\theta > 0$ and all $p_i = 0$. From the tight constraint, $q_i s_{i} = \theta$ for all active nodes. This perfectly equalizes the contributions ($c_i = q_is_i$) across the active set $A$. To maximize $R_\theta$ for any set size $|A|$, the denominator $\sum 1/s_{i}$ must be minimized. Therefore, the optimal active set strictly consists of the $|A|$ nodes with the highest input inner opinions $s_{i}$ (\textbf{Case 2: Equalization}).
\end{enumerate}

Because any hybrid solution (where both $\theta > 0$ and $p_i > 0$) yields a strictly sub-optimal knapsack return no intermediate solutions exist, concluding the proof.

\end{proof}

\subsection{Proof of Lemma \ref{lemma:ellip_closed_form}: Uncorrelated Ellipsoidal Noise}

\begin{proof}
Let $N = \sqrt{\sum_{j=1}^n (q_j s_{j} \sigma_j)^2}$. We formulate the Lagrangian for the optimization problem~\eqref{eq:ellip_robust}, introducing $\lambda$ for the constraint $\sum q_i = 1$, and $\mu_i \ge 0$ for the non-negativity constraints $q_i \ge 0$:
\begin{equation}
    \mathcal{L}(\mathbf{q}, \lambda, \boldsymbol{\mu}) = \sum_{i=1}^n q_i s_{i} - \omega N - \lambda \left( \sum_{i=1}^n q_i - 1 \right) + \sum_{i=1}^n \mu_i q_i
\end{equation}
The KKT stationarity condition requires $\frac{\partial \mathcal{L}}{\partial q_i} = 0$ for all $i$:
\begin{equation}
\label{eq:kkt_stationarity_app}
    s_{i} - \omega \frac{q_i s_{i}^2 \sigma_i^2}{N} - \lambda + \mu_i = 0 \implies \omega \frac{q_i s_{i}^2 \sigma_i^2}{N} = s_{i} - \lambda + \mu_i
\end{equation}
By complementary slackness, $\mu_i q_i = 0$ for all $i$. This divides the nodes into two disjoint sets:
\begin{itemize}
    \item \textbf{Inactive Nodes ($q_i = 0$):} Since $q_i = 0$, Eq.~\eqref{eq:kkt_stationarity_app} reduces to $\mu_i = \lambda - s_{i}$. The dual feasibility condition $\mu_i \ge 0$ requires that $s_{i} \le \lambda$.
    \item \textbf{Active Nodes ($q_i > 0$):} Here, $\mu_i = 0$. Eq.~\eqref{eq:kkt_stationarity_app} becomes $\omega \frac{q_i s_{i}^2 \sigma_i^2}{N} = s_{i} - \lambda$. Since $\omega, N, q_i$, and $(s_{i} \sigma_i)^2$ are strictly positive, the left-hand side is positive, strictly requiring $s_{i} > \lambda$.
\end{itemize}
This proves that the optimal solution only activates nodes with input inner opinions $s_{i}$ strictly larger than $\lambda$. Assuming a descending sort order for the $s_{i}$ values, the active set comprises exactly the top $k$ nodes. For an active node $i \le k$, isolating $q_i$ yields $q_i = \frac{N}{\omega} \frac{s_{i} - \lambda}{(s_{i} \sigma_i)^2}$. Applying the constraint $\sum_{i=1}^k q_i = 1$ allows us to eliminate the constants and find the normalized closed-form expression:
\begin{equation}
    q_i^* = \frac{ \frac{s_{i} - \lambda}{s_{i}^2 \sigma_i^2} }{ \sum_{j=1}^k \frac{s_{j} - \lambda}{s_{j}^2 \sigma_j^2} }
\end{equation}

To compute $\lambda$, we take the active node condition, divide by $s_{i} \sigma_i$, square the result, and sum over the active set. The $N^2$ terms perfectly cancel out, leaving the equation:
\begin{equation}
    \sum_{i=1}^k \frac{(s_{i} - \lambda)^2}{(s_{i} \sigma_i)^2} = \omega^2. 
\end{equation}
Expanding this sum reveals a standard quadratic form: $A_k \lambda^2 - 2B_k \lambda + C_k = 0$. 

While a quadratic equation generally yields two roots, $\lambda_{\pm} = \frac{B_k \pm \sqrt{B_k^2 - A_k C_k}}{A_k}$, the physical constraints of our problem strictly invalidate the positive root. We can show this by returning to the active node condition: $\frac{s_{i} - \lambda}{(s_{i} \sigma_i)^2} = \frac{\omega}{N} q_i$, and summing over all $k$ active nodes yields:
\begin{equation}
    \sum_{i=1}^k \frac{s_{i}}{(s_{i} \sigma_i)^2} - \lambda \sum_{i=1}^k \frac{1}{(s_{i} \sigma_i)^2} = \frac{\omega}{N} \sum_{i=1}^k q_i
\end{equation}
Substituting our definitions for $A_k$ and $B_k$, and applying $\sum_{i=1}^k q_i = 1$, this simplifies exactly to:
\begin{equation}
    B_k - \lambda A_k = \frac{\omega}{N}
\end{equation}
Because $\omega$ and $N$ are both strictly positive, the fraction $\frac{\omega}{N}$ must be strictly positive. Therefore, the system enforces a strict mathematical boundary: $B_k - \lambda A_k > 0$, which implies $\lambda < \frac{B_k}{A_k}$. 

By definition, $A_k$ is a sum of positive variances, so $A_k > 0$. The positive root of the quadratic is $\lambda_+ = \frac{B_k}{A_k} + \frac{\sqrt{B_k^2 - A_k C_k}}{A_k}$, which is strictly greater than $\frac{B_k}{A_k}$. If we were to accept $\lambda_+$, the term $(B_k - \lambda_+ A_k)$ would become negative. This would require the distance norm $N$ to be a negative value, which is impossible. Furthermore, a negative $\frac{\omega}{N}$ ratio would force the active influences to be negative ($q_i < 0$), fundamentally violating the non-negativity constraint of influences for the Friedkin-Johnsen model. 
Thus, only the negative root $\lambda_- = \lambda_k$ guarantees a valid distribution.
\end{proof}

\subsection{Comment for Algorithms \ref{alg:poly_solver} and \ref{alg:ellip_solver}}

We have provided pseudocode for the algorithms derived in Lemmas \ref{lm:poly_hetero_binary} and \ref{lemma:ellip_closed_form}. 
The dominant operation in both algorithms is the initial sorting step, after which the optimal target influence $\mathbf{q}^*$ is computed in linear time, yielding an overall $\mathcal{O}(n \log n)$ complexity. 
To prevent division-by-zero errors, the input domains are bounded. 
Specifically, $s_i \in (0,1]$, as nodes with $s_i = 0$ yield zero objective benefit and are safely filtered out. 
Similarly, we require $\boldsymbol{\sigma}^2 \in \mathbb{R}_{>0}^n$, since zero-variance nodes are deterministic and can be excluded from our uncertainty formulation.

\section{Proof of Theorem \ref{thm:exact_realization}}
\label{app:Feasibility}

To prove Theorem \ref{thm:exact_realization}, we first provide a geometric proof establishing the universal existence of a solution. We then provide a constructive algebraic proof that derives the exact closed-form equations used to efficiently compute this solution.

\subsection{Part 1: Geometric Proof of Universal Feasibility}

\begin{proof}
To establish the universal feasibility of any target influence distribution $\mathbf{q}^*$, we must prove that the function mapping stubbornness $\mathbf{a}$ to influence $\mathbf{q}$ is \textit{surjective}. In our context, surjectivity guarantees that for every mathematically valid target $\mathbf{q}$ in the output space, there exists at least one physically valid stubbornness vector $\mathbf{a}$ in the input space that perfectly achieves it.

Let $\Delta_{n-1} = \left\{ \mathbf{x} \in \mathbb{R}^n \;\middle|\; \sum_{i=1}^n x_i = 1, \, x_i \ge 0 \right\}$ denote the standard probability simplex. To prove part 1, we formulate the network dynamics as a continuous mapping between simplices.  

While the physical stubbornness vector $\mathbf{a}$ is bounded by the hypercube $[0,1]^n$, we strictly restrict our search space also to the simplex, such that $\mathbf{a} \in \Delta_{n-1} \subset [0,1]^n$. This domain restriction is mathematically sound: if we can prove that any target influence is achievable using a stubbornness vector from this restricted subset $\Delta_{n-1}$, it is necessarily achievable within the broader set $[0,1]^n$.

Following the geometric definitions provided by \cite{jamison1976factoring}, for any subset of indices $J \subset \{1, \dots, n\}$, a \textit{face} of the simplex is defined as the subset of vectors where the entries indexed by $J$ are exactly zero:
\begin{equation*}
    \Delta_{n-1}(J) = \{ \mathbf{x} \in \Delta_{n-1} \mid x_j = 0 \text{ for each } j \in J \}
\end{equation*}

To establish surjectivity, we invoke Lemma 2.1 from \cite{jamison1976factoring}:
\begin{quote}
\textbf{Lemma 2.1 (Surjectivity of Face-Preserving Maps):} If $f$ is a continuous map of $\Delta_{n-1}$ into itself which carries each face of $\Delta_{n-1}$ into itself ($f(\Delta_{n-1}(J)) \subseteq \Delta_{n-1}(J)$), then $f$ is surjective.
\end{quote}

In the Friedkin-Johnsen model, the resulting network influence vector $\mathbf{q}(\mathbf{a})$ is the column average of the influence matrix ${Q} = ({I} - ({I}-{A}){W})^{-1}{A}$, where ${A} = \text{diag}(\mathbf{a})$. The resulting influence vector is a probability distribution satisfying $\sum_{i=1}^n q_i = 1$ and $q_i \ge 0$. Therefore, the influence equation defines a mapping $f: \Delta_{n-1} \to \Delta_{n-1}$ such that $f(\mathbf{a}) = \mathbf{q}(\mathbf{a})$.

We now verify that this mapping $f$ satisfies the two conditions of Lemma 2.1:
\begin{enumerate}
    \item \textbf{Continuity:} By construction, the influence vector relies on basic matrix operations (addition, multiplication, and inversion). Because $\mathbf{a} \in \Delta_{n-1}$, we are guaranteed that $\sum_{i=1}^n a_i = 1$, implying that at least one node has strictly positive stubbornness ($a_k > 0$). For a strongly connected (irreducible) graph with row-stochastic matrix ${W}$, the matrix inverse $({I} - ({I}-{A}){W})^{-1}$ is mathematically guaranteed to be non-singular and well-defined everywhere on $\Delta_{n-1}$. Thus, the mapping $f(\mathbf{a})$ is strictly continuous over its entire domain.
    
    \item \textbf{Face-Preservation:} Suppose a stubbornness vector $\mathbf{a}$ lies on a face $\Delta_{n-1}(J)$. By definition, $a_j = 0$ for all $j \in J$. Nodes with zero stubbornness ($a_j=0$) guarantee zero influence ($q_j=0$) to the output vector $\mathbf{q}$. As such the input face $\mathbf{a} \in\Delta_{n-1}(J)$ also lies on the exact same face $\mathbf{q} \in\Delta_{n-1}(J)$. This confirms that $f(\Delta_{n-1}(J)) \subseteq \Delta_{n-1}(J)$, meaning the mapping carries each face into itself.
\end{enumerate}

Because our mapping $f(\mathbf{a})$ satisfies both continuity and face-preservation, Lemma 2.1 dictates that its image must cover the entire target simplex $\Delta_{n-1}$. Therefore, for any theoretically optimal influence distribution $\mathbf{q}^* \in \Delta_{n-1}$, there exists at least one physically valid stubbornness vector $\mathbf{a}^* \in \Delta_{n-1} \subset [0,1]^n$ that strictly achieves $f(\mathbf{a}^*) = \mathbf{q}^*$.
\end{proof}

\subsection{Part 2: Constructive Proof}

\begin{proof}

In the Friedkin-Johnsen model, the influence matrix is defined as ${Q} = ({I} - ({I}-{A}){W})^{-1}{A}$, where ${A} = \text{diag}(\mathbf{a})$. The influence vector $\mathbf{q}^*$ is:

\begin{equation}
(\mathbf{q}^*)^T = \frac{1}{n} \mathbf{1}^T {Q} = \frac{1}{n} \mathbf{1}^T ({I} - ({I}-{A}){W})^{-1} {A}
\end{equation}

We define the row vector $\mathbf{y}^T$:

\begin{equation}
\label{eq:aux_y}
\mathbf{y}^T = \frac{1}{n} \mathbf{1}^T ({I} - ({I}-{A}){W})^{-1}
\end{equation}

Substituting this back into the influence equation yields a direct, element-wise relationship between the target influence, the vector $\mathbf{y}$, and the unknown stubbornness:

\begin{equation}
\label{eq:decoupled_a}
(\mathbf{q}^*)^T = \mathbf{y}^T {A} \implies \mathbf{q}^* = {A}\mathbf{y} \implies q_i^* = a_i y_i
\end{equation}

From Eq.~\eqref{eq:decoupled_a}, it is strictly required that if $q_i^* = 0$, the corresponding stubbornness must be $a_i = 0$. For active nodes where $q_i^* > 0$, we must solve for $y_i$.

We multiply both sides of Eq.~\eqref{eq:aux_y} from the right by the inverse matrix term and transpose the entire equation to obtain a column-vector system:

\begin{equation}
({I} - {W}^T({I}-{A}))\mathbf{y} = \frac{1}{n} \mathbf{1}
\end{equation}

Expanding the left side yields $\mathbf{y} - {W}^T\mathbf{y} + {W}^T{A}\mathbf{y} = \frac{1}{n} \mathbf{1}$. Applying the equation: ${A}\mathbf{y} = \mathbf{q}^*$, matrix ${A}$ is entirely eliminated, leaving a sparse linear system:

\begin{equation}
({I} - {W}^T)\mathbf{y} = \frac{1}{n} \mathbf{1} - {W}^T\mathbf{q}^*
\end{equation}

Let ${C} = ({I} - {W}^T)$ and $\mathbf{b} = \frac{1}{n} \mathbf{1} - {W}^T\mathbf{q}^*$. Because ${W}$ is row-stochastic, ${W}\mathbf{1} = \mathbf{1}$, ${C}$ is a \emph{singular} matrix whose null space is spanned by the stationary distribution $\boldsymbol{\pi}$ (where ${C}\boldsymbol{\pi} = \mathbf{0}$). By the Fredholm Alternative, a solution to ${C}\mathbf{y} = \mathbf{b}$ exists if and only if the null space of $C^T$, which is: $\mathbf{1}$, because: $C^T\mathbf{1} = ({I} - {W})\mathbf{1} = \mathbf{0} $, is orthogonal with $\mathbf{b}$, i.e., $\mathbf{1}^T \mathbf{b} = 0$. Evaluating the sum:

\begin{equation}
\mathbf{1}^T \mathbf{b} = \mathbf{1}^T \left( \frac{1}{n} \mathbf{1} - {W}^T\mathbf{q}^* \right) = 1 - (\mathbf{1})^T \mathbf{q}^* = 1 - \mathbf{1}^T\mathbf{q}^* = 0
\end{equation}

Because the target influence distribution $\mathbf{q}^*$ sums to exactly 1, the system is perfectly consistent. Using any sparse linear solver, we can find a base solution $\mathbf{y}_0$. The complete family of valid solution vectors for equation $C \mathbf{y = \mathbf{b}}$ is given by: $\mathbf{y} = \mathbf{y}_0 + \eta \boldsymbol{\pi}$, for any scalar $\eta \in \mathbb{R}$.

Using Eq.~\eqref{eq:decoupled_a}, the stubbornness is given by $a_i = \frac{q_i^*}{y_{0,i} + \eta \pi_i}$. To satisfy the upper boundary constraint $a_i \le 1$ for all active nodes, the denominator must satisfy $y_{0,i} + \eta \pi_i \ge q_i^*$. Because the graph is strongly connected, the Perron-Frobenius theorem guarantees $\pi_i > 0$, for all nodes, allowing us to isolate $\eta$:

\begin{equation}
\eta \ge \frac{q_i^* - y_{0,i}}{\pi_i}
\end{equation}

To ensure that: $a_i \in [0,1] \quad \forall i$, across the entire network while achieving the target influence distribution $\mathbf{q^*}$, we select the scalar shift $\eta$ that satisfies: 

\begin{equation}
\eta = \max_{i : q_i^* > 0} \left( \frac{q_i^* - y_{0,i}}{\pi_i} \right)
\end{equation}

By definition of the maximum, by setting this $\eta$ guarantees that $y_{0,i} + \eta \pi_i \ge q_i^*$ for every active node $i$. Because active nodes strictly require $q_i^* > 0$, it immediately follows that the denominator is strictly positive, i.e., $y_{0,i} + \eta \pi_i \ge q_i^* > 0$. Consequently, the stubbornness fraction $a_i^* = \frac{q_i^*}{y_{0,i} + \eta \pi_i}$ is guaranteed to satisfy both $a_i^* \le 1$ and $a_i^* > 0$. Substituting $\eta$ yields the optimal, bounded stubbornness $a_i^*$ values in Eq.~\eqref{eq:optimal_stubbornness}, completing the proof.

\end{proof}

\section{Minimizing Stubbornness Change for Step 2}
\label{app:ETA}

After an optimal target influence $\influencetarg$ is found in Step 1, Theorem~\ref{thm:exact_realization} reveals that the set of all stubbornness vectors that can achieve $\influencetarg$ are parameterized by a scalar $\eta$. 
We can leverage this to find the stubbornness vector $\stubbotarg$ that minimizes the deviation from the initial stubbornness vector $\stubboInit$. 
We formulate this as minimizing the $L_2$ squared distance: $\|\stubbotarg - \stubboInit \|_2^2$. From here on out for notation simplicity we will write the initial stubbornness vector $\stubboInit$ as $\stubbo$. So we can minimize $\|\stubbotarg - \stubbo \|_2^2$ by solving:
\begin{equation}
\label{eq:budget_a}
    \eta^* = \arg\min_{\eta \ge \eta_{min}} \sum_{i : q_i^* > 0} \left( \frac{q_i^*}{y_{0,i} + \eta \pi_i} - a_{i} \right)^2
\end{equation}
Because this optimization problem is 1-dimensional and bounded, we can efficiently solve it to find $\eta^* \in [\eta_{min},+ \infty)$.
We have defined the stubbornness function for a single active node $i$ ($q_i^* > 0$) as:
\begin{equation}
    a_i^*(\eta) = \frac{q_i^*}{y_{0,i} + \eta \pi_i}
\end{equation}
The cost function over the set of active nodes $\mathcal{A}$, with $q_i>0$ we are trying to minimize is given by:
\begin{equation}
\label{eq:app_objective}
    f(\eta) = \sum_{i \in \mathcal{A}} \left( a_i^*(\eta) - a_i \right)^2
\end{equation}
subject to the strict lower bound $\eta \ge \eta_{min} = \max_{i \in \mathcal{A}} \left( \frac{q_i^* - y_{0,i}}{\pi_i} \right)$ ensures that $a_i^*(\eta) \le 1$ for all nodes as we have shown in Appendix \ref{app:Feasibility}.

\subsection{Analytical Limitations and Non-Convexity}

To find an analytical closed-form solution, we set the first derivative of Eq.~\eqref{eq:app_objective} to zero. The first and second derivatives of the individual stubbornness functions $a_i^*(\eta)$ with respect to $\eta$ are:
\begin{equation}
    \frac{\partial a_i^*}{\partial \eta} = -\frac{q_i^* \pi_i}{(y_{0,i} + \eta \pi_i)^2} \quad (< 0), \quad \text{and} \quad \frac{\partial^2 a_i^*}{\partial\eta^2} = \frac{2 q_i^* \pi_i^2}{(y_{0,i} + \eta \pi_i)^3} \quad (> 0)
\end{equation}
Applying the chain rule, the first derivative of the total cost function is: 

$$
f'(\eta) = 2 \sum_{i \in \mathcal{A}} (a_i^*(\eta) - a_i) \frac{\partial a_i^*}{\partial \eta} = 0.
$$

Solving this algebraically to find a closed-form solution is not possible. 
This is because it requires finding a common denominator across all $|\mathcal{A}|$ terms, which yields a polynomial of degree up to $3|\mathcal{A}| - 2$. 
Thus for any network with a non-trivial number of active nodes, there is no general algebraic solution for this equation.

We also analyze the convexity of $f(\eta)$ to determine if standard gradient-based optimization can guarantee global optimality. The second derivative of the objective is:
\begin{equation}
\label{eq:app_second_deriv}
    f''(\eta) = 2 \sum_{i \in \mathcal{A}} \Big[ \underbrace{\left(\frac{\partial a_i^*}{\partial \eta}\right)^2}_{> 0} + \underbrace{(a_i^*(\eta) - a_i)}_{\text{Can be } < 0} \underbrace{\frac{\partial^2 a_i^*}{\partial \eta^2}}_{> 0} \Big]
\end{equation}
If the optimal target influence requires a node to become strictly more open-minded than its initial state, we have $a_i^*(\eta) < a_i$. In such cases, the term $(a_i^*(\eta) - a_i)$ becomes strictly negative. If this negative magnitude is sufficiently large across the network, it can force $f''(\eta) < 0$. Thus, because the second derivative is not guaranteed to be positive, $f(\eta)$ is non-convex and may exhibit multiple local minima, rendering gradient methods unreliable.

\subsection{An Optimal \texorpdfstring{$\mathcal{O}(n)$}{O(n)} Algorithm}

Despite the non-convexity of the objective $f(\eta)$, the optimization domain is strictly one-dimensional and Lipschitz continuous, meaning its rate of change is bounded. We exploit this by designing a deterministic search space bounded by the structural limits of the stubbornness function. The lower bound, $\eta_{min}$, guarantees that $a_i^*(\eta) \le 1$. Conversely, because $a_i^*(\eta)$ is strictly monotonically decreasing, we establish a finite upper bound $\eta_{max}$ by forcing the target stubbornness of all active nodes below a small tolerance $\epsilon$, yielding:
$$\eta_{max} = \max_{i \in \mathcal{A}} \frac{q_i^* - \epsilon y_{0,i}}{\epsilon \pi_i}.$$ 

The global minimum is guaranteed to lie within the interval $[\eta_{min}, \eta_{max}]$.
To pinpoint this minimum, we partition the interval into a dense grid of $K$ uniformly spaced points and evaluate $f(\eta)$ on each. 
Evaluating the cost function $f$ requires operations over the active nodes $\mathcal{A}$, thus the entire grid search completes in $\mathcal{O}(K |\mathcal{A}|)$ time. 
Once the optimal grid point is identified, we isolate its immediate 2 neighbors and form a local bracket. 
Within this bracket, we apply Brent's method \cite{brent1973algorithms}, a root-finding algorithm combining bisection and inverse quadratic interpolation, to converge to the exact minimum $\eta^*$. 
Because $K$ is a fixed hyperparameter and $|\mathcal{A}| \le n$, the total time complexity strictly reduces to $\mathcal{O}(n)$, ensuring the method scales seamlessly to massive networks. 
Throughout our experimental evaluation we utilize this $\eta^*$ for the $\stubbotarg$ that achieves the  target influence $\influencetarg$.

\section{The LSMR Algorithm}
\label{app:LSMR}

In Theorem \ref{thm:exact_realization}, finding the base solution $\mathbf{y}_0$ requires solving the linear system:
\begin{equation}
    (\mathbf{I} - \mathbf{W}^T)\mathbf{y} = \frac{1}{n}\mathbf{1} - \mathbf{W}^T\mathbf{q}^*
\end{equation}

Let $\mathbf{C} = (\mathbf{I} - \mathbf{W}^T)$ and $\mathbf{b} = \frac{1}{n}\mathbf{1} - \mathbf{W}^T\mathbf{q}^*$. For real-world network applications, the interaction matrix $\mathbf{W}$ is typically large and highly sparse. Consequently, $\mathbf{C}$ is also large and sparse. Also, because $\mathbf{W}$ is row-stochastic, $\mathbf{C}$ is singular ($\mathbf{C}^T \mathbf{1} = \mathbf{0}$). Direct solvers, such as dense matrix inversion or LU factorization, are computationally intractable for large $n$ and destroy the sparsity of the graph. 

To compute $\mathbf{y}_0$ efficiently, we employ the \textbf{LSMR} (Least-Squares Minimum Residual) method \cite{fong2011lsmr}. LSMR is an iterative Krylov subspace method designed specifically for solving sparse linear equations and least-squares problems of the form:
\begin{equation}
    \min_{\mathbf{y}} \|\mathbf{C}\mathbf{y} - \mathbf{b}\|_2
\end{equation}

LSMR is based on the Golub-Kahan bidiagonalization process \cite{golub1965calculating}, which iteratively reduces the matrix $\mathbf{C}$ into a lower bidiagonal form. At each iteration $k$, the algorithm searches within the Krylov subspace:
\begin{equation}
    \mathcal{K}_k(\mathbf{C}^T\mathbf{C}, \mathbf{C}^T\mathbf{b}) = \text{span}\{ \mathbf{C}^T\mathbf{b}, (\mathbf{C}^T\mathbf{C})\mathbf{C}^T\mathbf{b}, \dots, (\mathbf{C}^T\mathbf{C})^{k-1}\mathbf{C}^T\mathbf{b} \}
\end{equation}

Crucially, LSMR operates without ever explicitly forming $\mathbf{C}^T\mathbf{C}$, since forming $\mathbf{C}^T\mathbf{C}$ severely degrades matrix sparsity. Instead, LSMR only requires computing matrix-vector products $\mathbf{C}\mathbf{v}$ and $\mathbf{C}^T\mathbf{u}$, making its memory requirements $\mathcal{O}(E)$, where $E$ is the number of edges in the graph.

We select LSMR over its popular predecessor, LSQR \cite{paige1982lsqr}, due to the specific algebraic properties of our model:
\begin{enumerate}
    \item \textbf{Exact Solution for Consistent Systems:} As proven in Appendix \ref{app:Feasibility}, our singular system $\mathbf{C}\mathbf{y} = \mathbf{b}$ is strictly consistent ($\mathbf{1}^T \mathbf{b} = 0$). Therefore, LSMR does not merely find an approximate "best-fit"; the least-squares residual $\|\mathbf{C}\mathbf{y} - \mathbf{b}\|_2$ converges to exactly zero.
    \item \textbf{Monotonic Convergence:} Matrices derived from graph Laplacians or $\mathbf{I} - \mathbf{W}^T$ are often poorly conditioned. While LSQR minimizes $\|\mathbf{C}\mathbf{y} - \mathbf{b}\|_2$, it can exhibit highly erratic behavior in the residual of the normal equations $\|\mathbf{C}^T\mathbf{r}\|_2$, leading to premature termination in iterative solvers. LSMR strictly ensures that both $\|\mathbf{C}\mathbf{y} - \mathbf{b}\|_2$ and $\|\mathbf{C}^T\mathbf{r}\|_2$ decrease monotonically at every iteration \cite{fong2011lsmr}. This provides robust, deterministic convergence even for structurally difficult graph topologies.
\end{enumerate}

By combining the sparse representations of SciPy \cite{virtanen2020scipy} with LSMR, our implementation isolates the mathematically exact $\mathbf{y}_0$ base solution. Because each iteration requires only sparse matrix-vector multiplications, the time complexity per iteration is $\mathcal{O}(E)$. 

\section{Extended Experimental Results}
\label{app:Experiements_More}

This section provides extended experimental results to support the experimental evaluation presented in Section~\ref{sec:experiments}.

\subsection{Extended RQ1 Results}

In Section~\ref{sec:experiments}, we demonstrated the robustness of our proposed solutions against worst-case noise using the \emph{Epinions} dataset. Here, we present the identical experimental setup evaluated on: \emph{Gowalla} (Figure~\ref{fig:rq1_extended_gowalla}), \emph{Google} (Figure~\ref{fig:rq1_extended_google}), and \emph{Pokec} (Figure~\ref{fig:rq1_extended_pokec}). 

Consistent with the main text, our \emph{Robust} solution continuously dominates the \emph{Nominal} and \emph{Initial} baselines across all graphs, maintaining high public opinion even under severe noise budgets ($\gamma, \omega$). Furthermore, the budgeted variant smoothly utilizes increasing intervention budgets ($b$) without suffering from the extreme volatility exhibited by the nominal strategy.

\begin{figure*}[ht] 
    \centering
    \includegraphics[width=0.5\textwidth]{Plots/rq1_shared_legend.pdf}

    \begin{subfigure}[b]{0.24\textwidth}
        \centering
        \includegraphics[width=\textwidth]{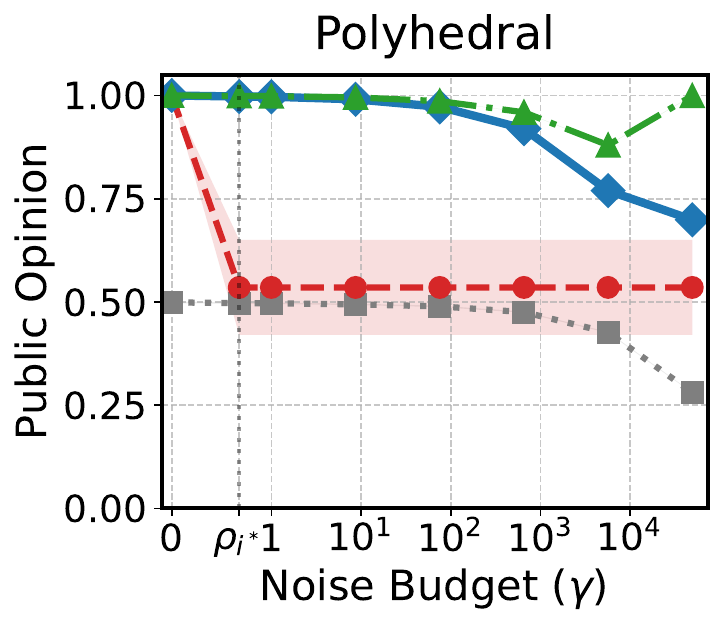}
        \caption{}
    \end{subfigure}\hfill
    \begin{subfigure}[b]{0.24\textwidth}
        \centering
        \includegraphics[width=\textwidth]{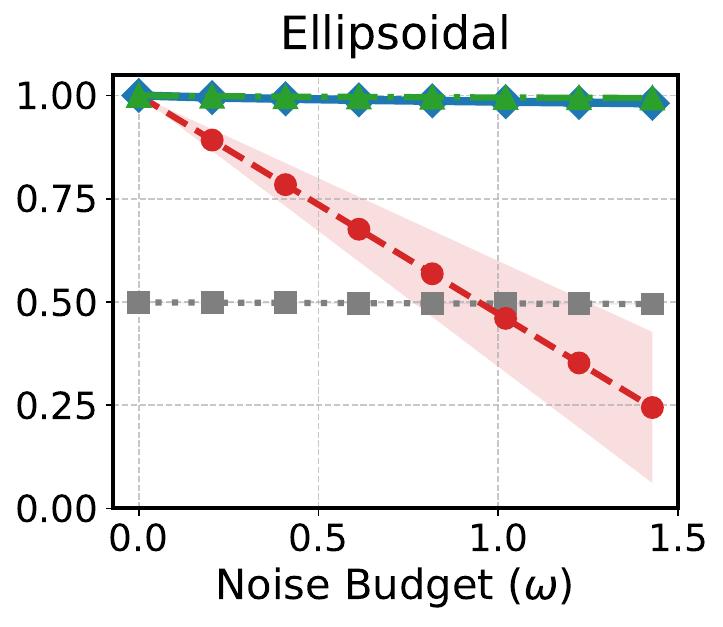}
        \caption{}
    \end{subfigure}\hfill
    \begin{subfigure}[b]{0.24\textwidth}
        \centering
        \includegraphics[width=\textwidth]{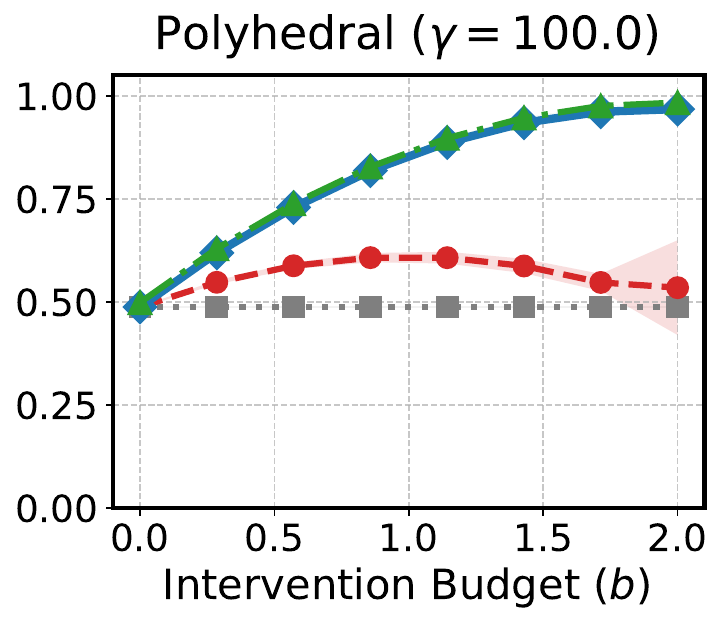}
        \caption{}
    \end{subfigure}\hfill
    \begin{subfigure}[b]{0.24\textwidth}
        \centering
        \includegraphics[width=\textwidth]{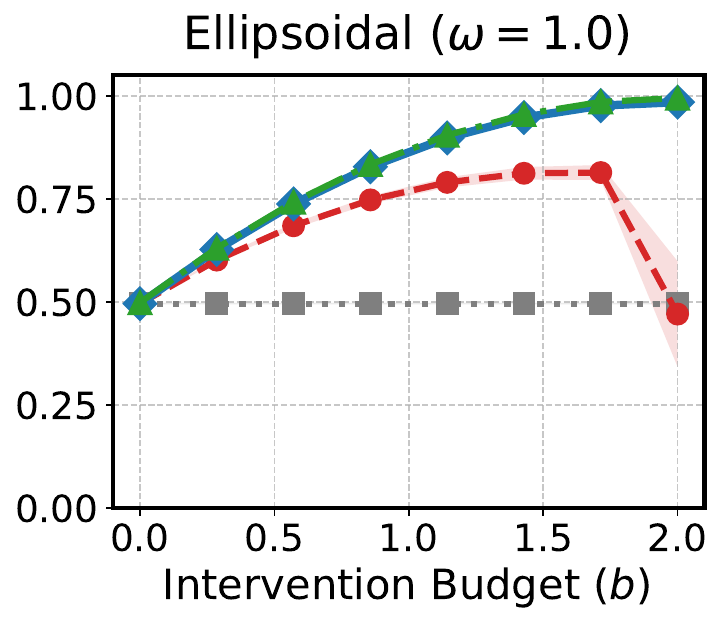}
        \caption{}
    \end{subfigure}
    
    \caption{Public opinion under varying noise and intervention budgets for the {Gowalla} dataset.}
    \label{fig:rq1_extended_gowalla}
\end{figure*}

\begin{figure*}[ht] 
    \centering
    \includegraphics[width=0.5\textwidth]{Plots/rq1_shared_legend.pdf}

    \begin{subfigure}[b]{0.24\textwidth}
        \centering
        \includegraphics[width=\textwidth]{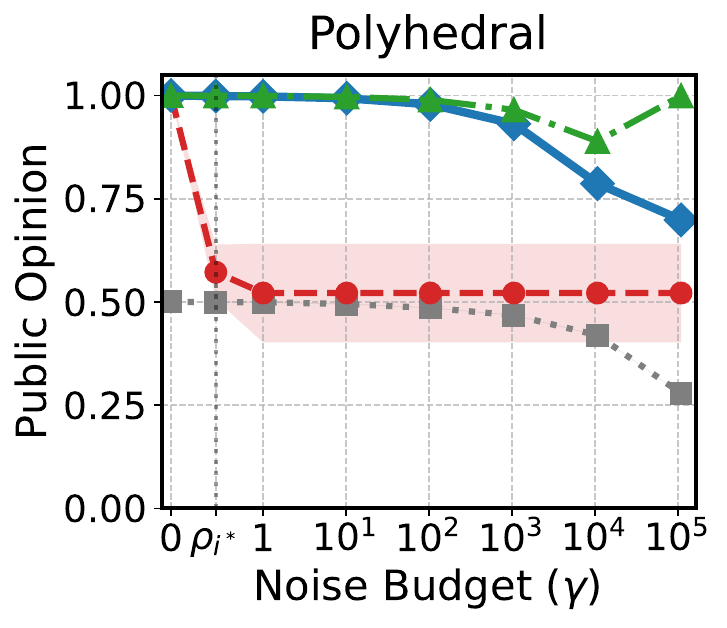}
        \caption{}
    \end{subfigure}\hfill
    \begin{subfigure}[b]{0.24\textwidth}
        \centering
        \includegraphics[width=\textwidth]{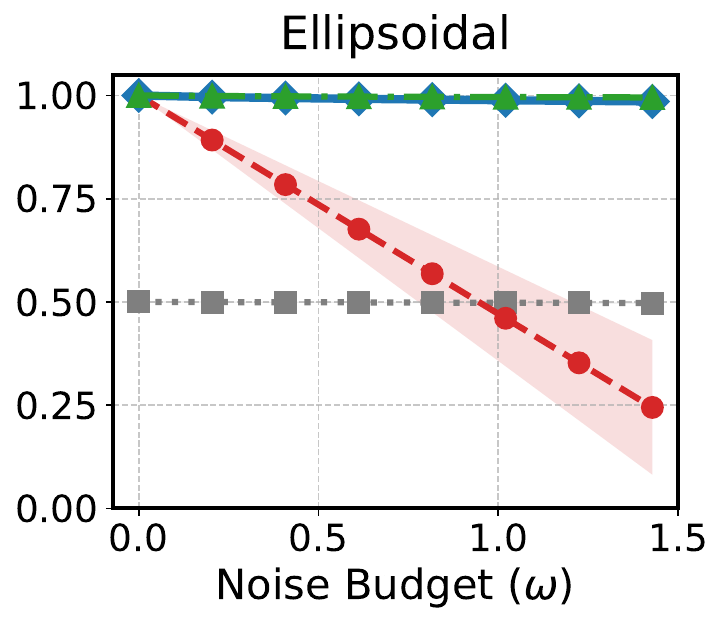}
        \caption{}
    \end{subfigure}\hfill
    \begin{subfigure}[b]{0.24\textwidth}
        \centering
        \includegraphics[width=\textwidth]{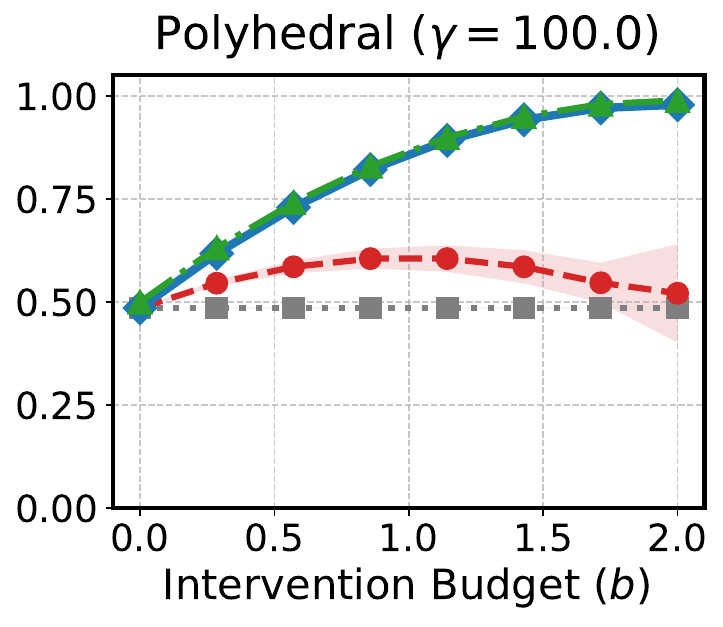}
        \caption{}
    \end{subfigure}\hfill
    \begin{subfigure}[b]{0.24\textwidth}
        \centering
        \includegraphics[width=\textwidth]{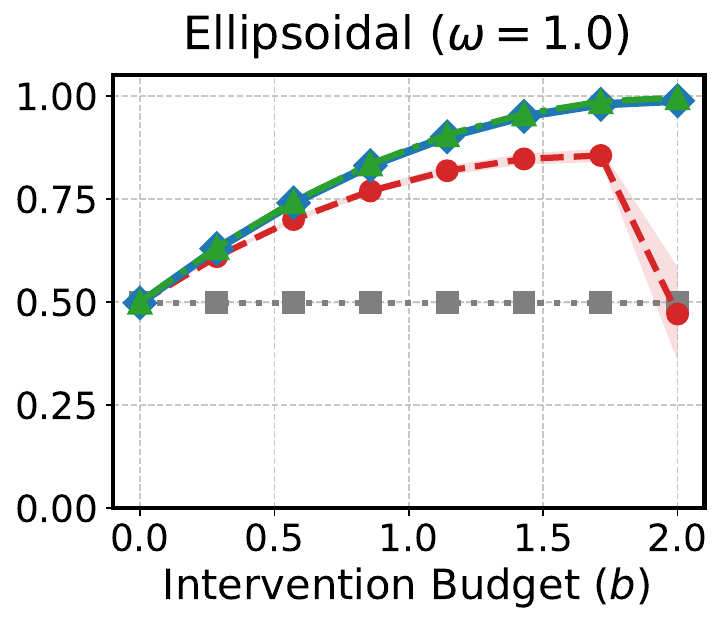}
        \caption{}
    \end{subfigure}
    
    \caption{Public opinion under varying noise and intervention budgets for the {Google} dataset.}
    \label{fig:rq1_extended_google}
\end{figure*}

\begin{figure*}[ht] 
    \centering
    \includegraphics[width=0.5\textwidth]{Plots/rq1_shared_legend.pdf}

    \begin{subfigure}[b]{0.24\textwidth}
        \centering
        \includegraphics[width=\textwidth]{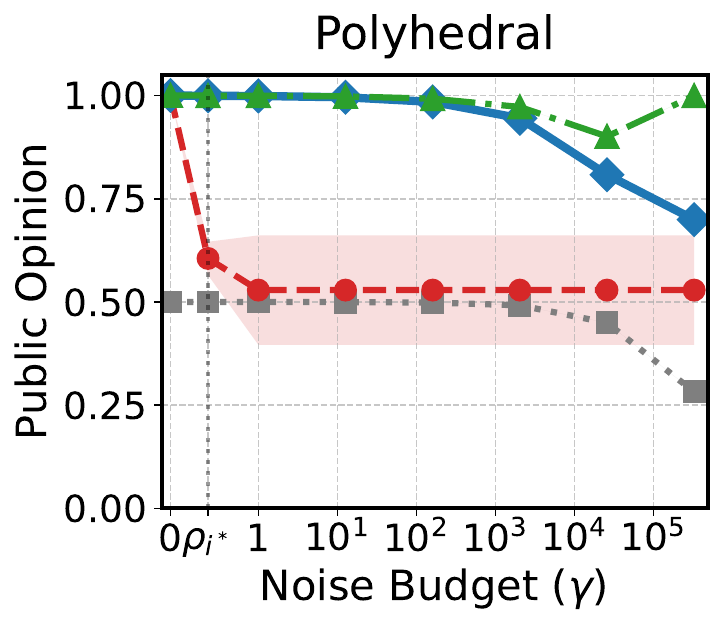}
        \caption{}
    \end{subfigure}\hfill
    \begin{subfigure}[b]{0.24\textwidth}
        \centering
        \includegraphics[width=\textwidth]{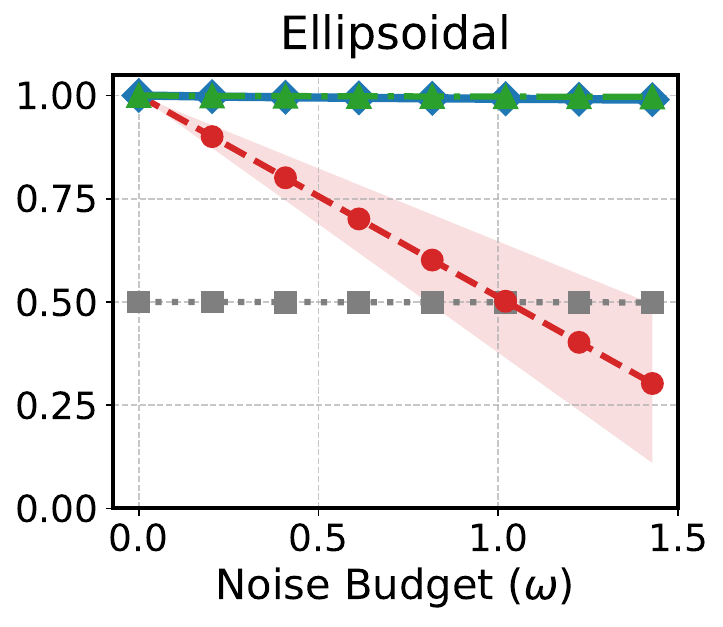}
        \caption{}
    \end{subfigure}\hfill
    \begin{subfigure}[b]{0.24\textwidth}
        \centering
        \includegraphics[width=\textwidth]{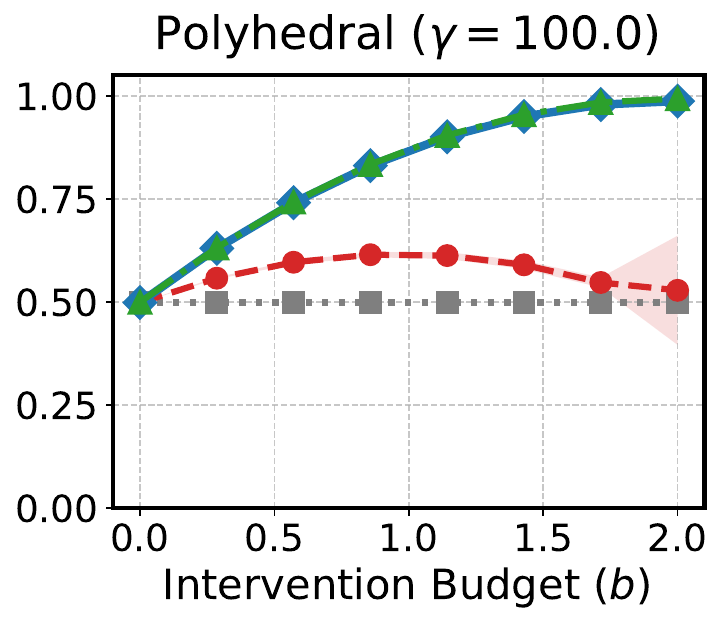}
        \caption{}
    \end{subfigure}\hfill
    \begin{subfigure}[b]{0.24\textwidth}
        \centering
        \includegraphics[width=\textwidth]{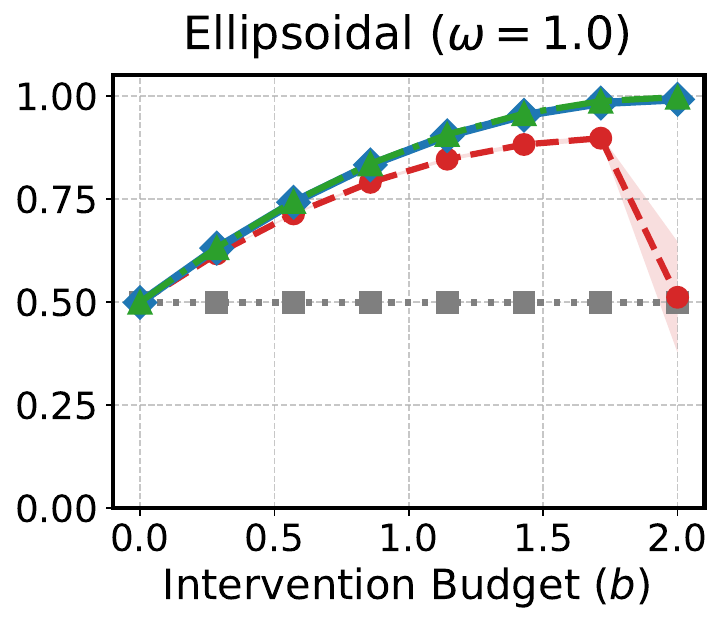}
        \caption{}
    \end{subfigure}
    
    \caption{Public opinion under varying noise and intervention budgets for the {Pokec} dataset.}
    \label{fig:rq1_extended_pokec}
\end{figure*}




\subsection{Extended RQ2 Results}
\label{app:extended_times}

In Section~\ref{sec:experiments}, we presented the average runtimes required to solve both the {\ProblemShort} and {\ProblemBudget} problems to address RQ2.
Here, in Table~\ref{tab:total_runtime_full}, we report the mean execution time (in seconds) alongside the standard deviation computed over 10 independent trials. 

\begin{table*}[ht]
    \centering
    \renewcommand{\arraystretch}{0.85} 
    \caption{Runtime (in seconds) across networks. Left section reports Step 1 of the {\ProblemShort} problem. Right section reports the two-step runtime for {\ProblemBudget} ($b=1$). Values represent the mean and one standard deviation across 10 independent trials.}
    \label{tab:total_runtime_full}
    \resizebox{\textwidth}{!}{%
    \setlength{\tabcolsep}{4pt} 
    \begin{tabular}{l r r r r r r r r r r} 
        \toprule
        \multirow{3}{*}{\textbf{Dataset}} & \multicolumn{4}{c}{\textbf{{\ProblemShort} Step 1}} & \multicolumn{6}{c}{\textbf{{\ProblemBudget} ($b=1$)}} \\
        \cmidrule(lr){2-5} \cmidrule(lr){6-11}
        & \multicolumn{2}{c}{\textbf{Polyhedral}} & \multicolumn{2}{c}{\textbf{Ellipsoidal}} & \multicolumn{3}{c}{\textbf{Polyhedral}} & \multicolumn{3}{c}{\textbf{Ellipsoidal}} \\
        \cmidrule(lr){2-3} \cmidrule(lr){4-5} \cmidrule(lr){6-8} \cmidrule(lr){9-11} 
        & Solver & \CustomLemmaThree & Solver & \CustomLemmaFour & Step 1 & Step 2 & \textbf{Total} & Step 1 & Step 2 & \textbf{Total} \\
        \midrule
        Epinions & $0.27 \pm 0.02$ & $\mathbf{0.0020 \pm 0.0002}$ & $0.46 \pm 0.03$ & $\mathbf{0.0021 \pm 0.0003}$ & $1.59 \pm 0.12$ & $4.98 \pm 0.35$ & $\mathbf{6.57 \pm 0.37}$ & $2.70 \pm 0.22$ & $3.78 \pm 0.28$ & $\mathbf{6.48 \pm 0.35}$ \\
        Gowalla  & $1.68 \pm 0.14$ & $\mathbf{0.0136 \pm 0.0012}$ & $2.92 \pm 0.21$ & $\mathbf{0.0148 \pm 0.0014}$ & $17.23 \pm 1.25$ & $33.47 \pm 2.80$ & $\mathbf{50.70 \pm 3.05}$ & $23.01 \pm 1.80$ & $19.65 \pm 1.40$ & $\mathbf{42.66 \pm 2.25}$ \\
        Google   & $3.92 \pm 0.32$ & $\mathbf{0.0332 \pm 0.0025}$ & $6.77 \pm 0.58$ & $\mathbf{0.0354 \pm 0.0028}$ & $77.88 \pm 5.80$ & $130.90 \pm 11.20$ & $\mathbf{208.78 \pm 12.60}$ & $100.57 \pm 8.50$ & $99.65 \pm 8.10$ & $\mathbf{200.22 \pm 11.75}$ \\
        Pokec    & $14.01 \pm 1.15$ & $\mathbf{0.1117 \pm 0.0085}$ & $20.01 \pm 1.75$ & $\mathbf{0.1168 \pm 0.0092}$ & $136.91 \pm 11.20$ & $278.72 \pm 24.50$ & $\mathbf{415.63 \pm 26.90}$ & $263.50 \pm 21.40$ & $229.33 \pm 18.60$ & $\mathbf{492.83 \pm 28.35}$ \\
        \bottomrule
    \end{tabular}%
    }
\end{table*}

\end{document}